\documentclass[journal,10pt]{IEEEtran}

\usepackage{amsthm,amsmath,amssymb,amsfonts}

\usepackage{dsfont}
\usepackage{bm}
\usepackage{booktabs}
\usepackage{siunitx}

\usepackage{pgfplots}
\usepackage{pdfpages}
\usepackage{color}

\usepackage{subcaption}
\usepackage{changes}

\usepackage{multirow}

\usepackage{paralist}

\newcommand{\I}{\mathds{1}}
\newcommand{\R}{\mathbb{R}}

\newtheorem{theorem}{Theorem}[section]

\newtheorem{prop}[theorem]{Proposition}

\DeclareMathOperator*{\argmin}{arg\,min}
\DeclareMathOperator*{\argg}{arg}

\usepackage[authoryear]{natbib}

\usepackage{authblk}

\graphicspath{{images/}}

\begin{document}
\title{Short term wind turbine power output prediction}

\author{S.~Kolumb\'an\IEEEauthorrefmark{2}, S.~Kapodistria\IEEEauthorrefmark{1}, and N.~Nooraee\IEEEauthorrefmark{3}
\IEEEcompsocitemizethanks{
\IEEEauthorrefmark{1} Department of Mathematics and Computer Science, Eindhoven University of Technology, P.O. Box 513, 5600 MB Eindhoven, the Netherlands. Email: s.kapodistria@tue.nl,\\
\IEEEauthorrefmark{2} Department of Mathematics and Computer Science of the Hungarian Line, Babeş-Bolyai University,\\
\IEEEauthorrefmark{3} Assoc. Principal Scientist at MSD}
}

\maketitle


\begin{abstract}
In the wind energy industry, it is of great importance to develop models that accurately forecast the power output of a wind turbine, as such predictions are used for wind farm location assessment or power pricing and bidding, monitoring, and preventive maintenance. As a first step, and following the guidelines of the existing literature, we use the {\em supervisory control and data acquisition} (SCADA) data to model the {\em wind turbine power curve} (WTPC). We explore various parametric and non-parametric approaches for the modeling of the WTPC, such as parametric logistic functions, and   non-parametric  piecewise linear, polynomial, or  cubic spline interpolation functions. We demonstrate that all aforementioned classes of models are rich enough (with respect to their relative complexity) to accurately model the WTPC, as their mean squared error (MSE) is close to the MSE lower bound calculated from the historical data. However, all aforementioned models, when it comes to forecasting, seem to have an intrinsic limitation, due to their inability to capture the inherent auto-correlation of the data. To avoid this conundrum, we  show that adding a properly scaled ARMA modeling layer increases short-term prediction performance, while keeping the long-term prediction capability of the model. We further enhance the accuracy of our proposed model, by incorporating additional environmental factors that affect the power output, such as the ambient temperature, and the wind direction.
\end{abstract}


\section{Introduction}
\IEEEPARstart{W}{ind}  turbine power curves (WTPC) are used for the modeling of the power output of a single wind turbine. Such models are needed in 
\begin{inparaenum}[i)]
\item  Wind power pricing and bidding: Electricity is a commodity which is traded similarly to stocks and swaps, and its pricing incorporates principles from supply and demand.
\item Wind energy assessment and prediction: Wind resource assessment is the process by which wind farm developers estimate the future energy production of a wind farm.
\item  Choosing a wind turbine: WTPC models aid the wind farm developers to choose the generators of their choice, which would provide optimum efficiency and improved performance.
\item Monitoring a wind turbine and for preventive maintenance: A WTPC model can serve as a very effective performance monitoring tool, as several failure modes can result in power generation  outside the specifications. As soon as an imminent failure is identified, preventive maintenance (age or condition based) can be implemented, which will reduce costs and increase the availability of the asset.
\item Warranty formulations: Power curve warranties are often included in contracts, to insure that the wind turbine performs according to specifications. Furthermore, service providers offer warranty and verification testing services of whether a turbine delivers its specified output and reaches the warranted power curve, while meeting respective grid code requirements.
\end{inparaenum}
See, e.g.,  \citep{Widen2015,Shi2011,Lydia2013} and the references therein for the aforementioned applications. Thus it is pivotal to construct accurate WTPC models. However, this is a difficult problem, as the output power of a wind turbine varies significantly with wind speed and every wind turbine has a very unique power performance curve,  \citep{Manwell2010}.

In this paper, we explore the literature on how to create an accurate WTPC model based on a real dataset and suggest practical and scientific improvements on the model construction. We initially construct a {\em static} model  (in which the regressor(s) and the regressand(s) are considered to be independent identically distributed (i.i.d.) random variables) model for the WTPC and demonstrate how various parametric and non-parametric approaches are performing from both a theoretical perspective, and also with regard to the data. In particular, we explore parametric logistic functions, and the  non-parametric  piecewise linear interpolation technique, the polynomial interpolation technique, and the cubic spline interpolation technique. We demonstrate that all aforementioned classes of models are rich enough to accurately model the WTPC, as their mean squared error (MSE) is close to a theoretical MSE lower bound. Within each non-parametric model class, we select the best model by rewarding MSEs close to the theoretical bound, while simultaneously penalizing for overly complicated models (i.e., models with many unknown parameters), using the Bayesian information criterion (BIC), see \citep{SchwarzBIC}. We demonstrate that such a static model, even after incorporating information on the wind speed and the available environmental factors, such as wind angle and ambient temperature, does not fully capture all available information. To this end, we propose in this paper, to enhance the static model with a {\em dynamic} layer (in which the regressor(s) is considered to be inter-correlated e.g., time series or stochastic processes), based on an autoregressive-moving-average (ARMA) modeling layer.

\subsection{Contribution of the paper}
In this paper, based on a real dataset, we explore a hybrid approach for the wind turbine power output modeling consisting of the static model plus the dynamic layer. This approach: 
\begin{inparaenum}[i)]
\item provides a very accurate modeling approach; 
\item is very useful for accurate short and long-term predictions; 
\item indicates that, within the cut-in wind speed  (\SI[per-mode=symbol]{3.5}{\meter\per\second}) and the rated output wind speed (\SI[per-mode=symbol]{15}{\meter\per\second}), the conditional distribution of the power output is Gaussian. 
\end{inparaenum}
We consider that points i)-ii) mentioned above will contribute directly to the practice, as the accurate modeling and forecasting capabilities are of utter importance. We follow the straightforward and industrially accepted approach of first estimating the WTPC, then extending this model with additional modeling layers. Although, there are various ways of modeling the WTPC, we argue, that above a given level of modeling flexibility, the exact choice is not important. The presented approach relies on available predictions of the environmental conditions. Furthermore, point \emph{iii)} mentioned above will greatly benefit the literature, as it is the first stepping stone towards proving that random power injections from wind energy in the electric grid can be accurately modeled using a Gaussian framework, see \citep{bert1,bert2}.
All in all, in this paper, we provide a new dataset collected from a wind turbine, and use it to show how to accurately model and forecast power output. Since the features of the used SCADA data are common among other types of turbines, the conclusions of the paper remain valid for other turbines as well. The analysis presented in this paper is scientifically and practically relevant, and contributes substantially from both the modeling and forecasting aspect, while providing a thorough overview of sound statistical methods. All results presented in the paper are motivated scientifically (when appropriate and possible) and are supported by real data.

\subsection{Paper outline}
In Section \ref{sec:Data}, we describe the raw data and provide all information on how the data was cleaned. In Section \ref{sec:power_curve_modeling}, we treat the WTPC modeling: First, in Section \ref{sec:mod_classes}, we present a simple static WTPC modeling approach, which models the power output as a function of the wind speed, using both parametric and non-parametric approaches; parametric logistic models (Section~\ref{sec:Logistic}), non-parametric piecewise linear models (Section~\ref{sec:PiecewiseLinear}), polynomial models (Section~\ref{sec:Polynomial}), and spline models (Section~\ref{sec:Spline}).
In Section \ref{sec:ComparisonOfModelClasses}, we compare the various modeling classes and determine criteria for model selection. In Section \ref{sec:MorePhysicalParamaters}, we enhance the static model at hand by incorporating additional factors, such as the wind angle and the ambient temperature. Analyzing the residuals of the enhanced static model, we are motivated to introduce a dynamic Gaussian layer in our model, cf. Section \ref{sec:TimeSeriesPrediction}, which produces very accurate short-term predictions, cf. Section \ref{sec:forecasting_confidence}. We conclude the paper with some remarks in Section \ref{sec:conclusions}.

\begin{figure}[h!]
\centering
\includegraphics[scale=0.5]{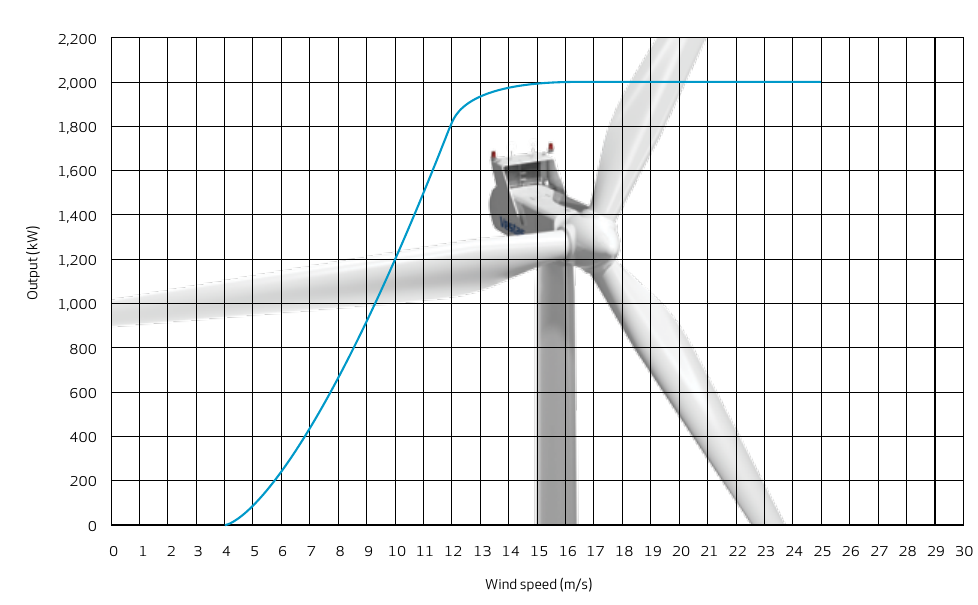}
\caption{WTPC of the  V80-2.0MW (picture from \cite{Vestas})}
\label{img:WTPC-V80}
\end{figure}

\section{Data}
\label{sec:Data}
The goal of this section is to describe the features of the data used in this study. The data was obtained by the supervisory control and data acquisition (SCADA) system of a wind turbine operator in the Netherlands.
The data was collected from an off-shore Vestas V80-2.0MW wind turbine, with a rated capacity of \SI[per-mode=symbol]{2}{\mega\watt}. Vestas V80-2.0MW joins the grid connection at a wind speed of \SI[per-mode=symbol]{4}{\meter\per\second}, has a rated actual power output of \SI[per-mode=symbol]{2}{\mega\watt} (typically achieved) at a wind speed of \SI[per-mode=symbol]{16}{\meter\per\second}, and it is disconnected at a wind speed of \SI[per-mode=symbol]{25}{\meter\per\second}. See
Figure \ref{img:WTPC-V80} for  a depiction of the theoretical WTPC. These are suggested values offered by the manufacturer, but they might change due to wear of the turbine or due to  installation or geographic circumstances.\\
The data used for the analysis presented in this paper spans across two years and the dataset contains recordings of the environmental conditions, as well as the physical state, and power output of the turbine.\\
There are two important features of SCADA datasets, which are not specific to the data of our study but are common amongst SCADA datasets recorded throughout the wind industry. One of these is the \SI{10}{\minute} reported frequency of the SCADA observations; although the signals of interest are collected at a relatively high frequency, only processed observations calculated on a \SI{10}{\minute} window are recorded in the SCADA databases. These processed signals contain the average, maximum, minimum and standard deviation of the wind speed, and the power output amongst other quantities of interest.
The second important feature is that the data are strongly quantized due to the rounding of the reported number.   As a result, the observations are recorded up to one decimal digit.
Some of our finding are consequences of these two properties which correspond to the quasi industry standard. Because of this, we expect that our results are also applicable to similar data coming from other wind turbine operators or wind turbine service providers.

\subsection{Description of the raw dataset}
All graphs and figures were produced using two seasonal parts of the available dataset. Throughout the paper, we  refer to the data recorded between June 1, 2013, and August 31, 2013, as the training data, and the corresponding period of year 2014 as the validation data. Although, we have access to the full two year data set, we choose to restrict our analysis in a specific season of the year, as this reduces seasonality effects, while still maintaining a significant amount of data, and it permits a full decoupling between the training and the validation data. It is important to note that the results presented in the paper still hold when we perform the same analysis using the full year 2013 as training data and the data from 2014 for validation. \\
The dataset contains observations of various signals every 10 minutes. Some of the signals contained in the dataset are the ambient wind speed, say $w_t$, the relative direction of the wind speed with respect to the nacelle, say $\phi_t$, the ambient temperature, say $T_t$, and the power output produced by the turbine, say $p_t$, at time $t$, $t\geq0$. 
Besides the aforementioned continuous valued signals, there are some nominal variables with a discrete support, such as the variable pertaining to the different operational states of the turbine. Such variables help to identify time periods during which the turbine is out of use (maintenance, free run, blades turned into low resistance position) or if the wind turbine is in a state different from normal operational condition.

In the first part of the paper, we suppress the subscript $t$ as we deal with static models, while in the second part of the paper we deal with dynamic models, and we, therefore, reinstate the subscript $t$ notation.

\begin{table}[b!]
\caption{Summary report of the data cleaning}
\label{tab:DataCleaningResults}
\begin{center}
\begin{tabular}{  l  c c  }
 \hline
 \bf{\parbox{5.5cm}{Year}} & \bf{2013} & \bf{2014} \\
 \hline
{\parbox{5.5cm}{Total number of observations}} & 13248 & 13248 \\
{\parbox{5.5cm}{Number of NAs, IN, \& NNO observations}}& 255  & 445 \\
{\parbox{5.5cm}{Number of outliers}} & 144 & 165 \\
{\parbox{5.5cm}{Total number of observations after cleaning }}& 12849 & 12638 \\
{\parbox{5.5cm}{Percentage of observations kept after cleaning}} & 97 \% & 95.4 \% \\
 \hline
\end{tabular}
\end{center}
\end{table}

\subsection{Cleaning the data}
The quality of the available SCADA data is extremely good, nevertheless it requires some pre-processing before creating the forecasting models. We list below the cleaning rules implemented in this study, according to which we disregard observations:
\begin{enumerate}
\item Missing entries (NAs): there are a few timestamps that are completely missing from the 10 minute sampling sequence.
\item Incomplete entries (IN): if one or more of the signal values, e.g., the power output, the wind speed, etc., are missing from a data record, then the full record corresponding to this time stamp is discarded.
\item Not normal operation (NNO): based on the value of the state variables we can disregard states that do not correspond to normal operational conditions, e.g. free rotation of the wind turbine without connection to the grid, derated operation, etc.
\item Outliers: Firstly, all observations of wind power corresponding to the same wind speed are grouped together and the corresponding box plot is generated. Then, for every given wind speed value, all points with power generation  outside the whiskers of the box plot (i.e., all observations falling outside the interval $(Q_1-3\text{IQR},Q_3+3\text{IQR})$) are discarded.
\end{enumerate}

Table~\ref{tab:DataCleaningResults} contains the summary report of the data cleaning procedure. It shows that approximately 5\% of the original data is discarded, still leaving a trove of data to be used for estimation purposes. The scatter plot of the power output, $p$, against the wind speed, $w$, is shown in Figure~\ref{fig:2013_full-clean-overlayFullClean}.
In this figure, we have color-characterized the training data by depicting in red the raw data, and in blue the cleaned dataset used for the analysis.

\section{Power curve modeling and its limitations}
\label{sec:power_curve_modeling}
When it comes to pricing wind power, assessing the possible location for a wind turbine installation or to forecasting short-term (expected) power generation  for supply purposes, the main tool proposed in the literature is the WTPC. Such a curve is used to  describe the relationship between the steady wind speed and the produced power output of the turbine. The shape of the WTPC for the type of wind turbines of interest to this study is depicted in Figure \ref{img:WTPC-V80}, while the fitted curve based on the cleaned data is depicted in Figure~\ref{fig:2013_full-clean-overlayCleanWTPC}.

\begin{figure}[h!]
	\centering
	\includegraphics{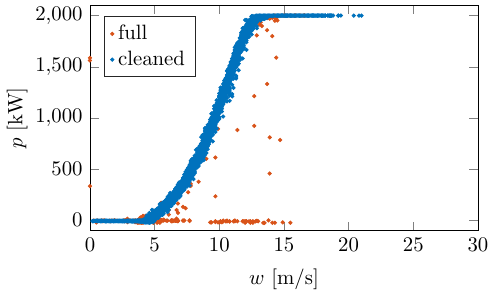}	
	\caption{The power output, $p$, against the wind speed, $w$, in the raw (red) and the cleaned (blue) dataset from 2013}
	\label{fig:2013_full-clean-overlayFullClean}
\end{figure}
\begin{figure}[h!]
	\centering
	\includegraphics{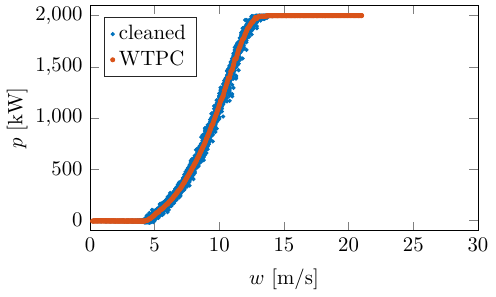}
	\caption{The power output $p$ against wind speed $w$ for the cleaned training dataset together with an estimated WTPC}
	\label{fig:2013_full-clean-overlayCleanWTPC}
\end{figure}

The WTPC, in ideal (laboratory) conditions, is given by the International Standard or the manufacturer, cf. \citep{IEC,Vestas}, but such curves can change over time due to environmental changes or due to component wear. This makes it paramount to estimate the power curve for each turbine individually, so these tailor-made WTPCs may be used for power generation  forecasting, decision under uncertainty, and monitoring.

The literature dealing with the topic of WTPC modeling techniques is extensive and covers many fields, see, e..g., \citep{Li2001,Lydia2014452,CarrilloObando2013,Lydia2013,Sohoni2016} and the references therein. The majority of this work is  focused on obtaining the best parametric or non-parametric model for the power curve of a turbine based on the available data. To this purpose, different approaches are compared using various criteria.  However, as we point out in this paper by comparing the various  model classes proposed in the literature, obtaining a parametric or non-parametric estimate of the WTPC is of limited value. 
Although Figure~\ref{fig:2013_full-clean-overlayCleanWTPC} indicates an ``appropriate'' fitted model to the power output using the wind speed, there are apparent remaining residuals that are not explained by the fitted power curve.  Investigating the statistical properties of the residuals reveals features that should be taken into account in the modeling.
The remaining residuals can be explained by the fact that it might be needed to use additional regressors (besides the wind speed) to explain the power output, and can be attributed to the fact that the homoscedasticity assumption is not valid, i.e. the variance of the residuals is not constant and the residuals are highly correlated. 
These two issues can be partly overcome by considering machine learning and artificial intelligence approaches, see, e..g., \citep{Peletier2016,Lazaro2022,Tumsea2022} and the references therein. However, the drawback in these approaches is that the models are not interpretable, that they require a trove of data for their training, and that they cannot be easily transferred to model the WTPC of another wind turbine. 
For these reasons, we strongly believe that, contrary to the existing literature,  the focus should shift to obtaining models that are interpretable, that can be easily extended to various regressors, and that can capture the heteroscedastic nature of the data. Such models  should not only be ranked according to the regular modeling power, but also according to their computational complexity and their numerical robustness.

\subsection{Power curve modeling classes}
\label{sec:mod_classes}

In this section, we present and compare various model classes proposed in the literature. Firstly, we introduce some notation in Section~\ref{sec:ModelingAndEstimation} that allows us to describe in a uniform manner the  models belonging to different model classes.
Thereafter, we describe how to calculate the estimates for each model class. For all model classes, we assume that the value of the power curve is constant below \SI[per-mode=symbol]{3.5}{\meter\per\second} taking the estimated power output value at \SI[per-mode=symbol]{3.5}{\meter\per\second}. Similarly, in the wind speed range between \SI[per-mode=symbol]{15}{\meter\per\second} and \SI[per-mode=symbol]{25}{\meter\per\second}, the power output curve is constant taking the estimated power output value at  \SI[per-mode=symbol]{15}{\meter\per\second}. Above the cut-out speed \SI[per-mode=symbol]{25}{\meter\per\second}, the power output is set to zero as the turbine should not operate. Thus, this part of the curve is not estimated, as in addition the cleaned dataset does not contain any observations in this range. These limitations need to be incorporated into the estimation procedure of the specific models, the details of which are presented in Section~\ref{sec:IncorporatingBounds}.
We note that the chosen \SI[per-mode=symbol]{15}{\meter\per\second} as an upper bound, for the rated speed, is lower than the value suggested by the manufacturer (see Section \ref{sec:Data}), but the data validate our choice.
\\
We consider both parametric and non-parametric models and compare the various model classes using the mean squared error (MSE) value, while within a class (for the non-parametric models) we select a model taking into account the complexity associated with it; non-parametric model classes (e.g. polynomial or spline models) have a nested structure, where the nesting levels correspond to complexity levels within the class (e.g. degree of polynomial or knot points of splines). The selection procedure  of a model within a class  is presented in Section~\ref{sec:ModelSelection}.

\subsubsection{Modeling and least squares estimation}
\label{sec:ModelingAndEstimation}

A power curve is a functional relation between the wind speed, $w$, and the power generation,  $p$. We define this functional relation as
\begin{equation}
p = \mathcal{M}_{\bm{\theta}}(w),
\end{equation}
where $\mathcal{M}_{\bm{\theta}}(\cdot)$ is a function belonging to the model class parametrized by a vector of parameters $\bm{\theta} \in \R^{n_{\bm{\theta}}}$, with $n_{\bm\theta}$ the dimension of the parameter vector depending on the model class $\mathcal{M}$. In the next sections, we consider various parametric and non-parametric model classes: logistic models $\mathcal{G}_{\bm\theta}(\cdot)$ in Section~\ref{sec:Logistic}, piecewise linear models $\mathcal{L}_{\bm\theta}(\cdot)$ in Section~\ref{sec:PiecewiseLinear}, polynomial models $\mathcal{P}_{\bm\theta}(\cdot)$ in Section~\ref{sec:Polynomial}, and spline models $\mathcal{S}_{\bm\theta}(\cdot)$ in Section~\ref{sec:Spline}.

Given a dataset containing a series of wind speed and power output pairs, $(w_k, p_k)_{1 \leq k \leq N}$ with $N$ the total number of observations, we define the least squares estimate within a model class as
\begin{equation}
\hat{\bm\theta} = \argmin_{\bm \theta} \frac{1}{N} \sum_{k=1}^N (p_k - \mathcal{M}_{\bm \theta}(w_k))^2.
\end{equation}
In order to shorten notation, the power output given by the least-squares estimated model at a given time is going to be denoted as $\hat{p} = \mathcal{M}_{\hat{\bm \theta}}(w)$, where the model class $\mathcal{M}$ is always going to be clear from the context.

One important conclusion of the paper is that WTPC modeling has significant limitations. In order to show this, we introduce some elementary facts about least squares estimates related to quantized data (as the SCADA data are quantized to one decimal digit, as described in Section~\ref{sec:Data}).

\begin{prop}[Lower bound for MSE]
\label{stat:MSELowerBound}
Irrespective of the model structure that is used to fit a model to the training data, if the training data are quantized in the regressor then there is a minimal attainable MSE and that can be calculated based on the data.

Let the samples be $\left(x_k, y_k \right)_{1 \leq k \leq N}$ and consider a model ${y}=\mathcal{M}_{{\theta}}(x)$, with  least squares estimate
$$
\hat{\theta} = \argmin_{\theta} \frac{1}{N} \sum_{k=1}^N (y_k - \mathcal{M}_\theta(x_k))^2.
$$
Let $\mathcal{X}$ be the set of all appearing values of $x$, i.e. $\mathcal{X} = \bigcup_{k=1}^N \{ x_k \}$, then the minimal attainable MSE value can be calculated as
\begin{equation}
\label{eq:MSELowerBound}
\min_{\theta} \frac{1}{N} \sum_{k=1}^N (y_k-\mathcal{M}_\theta(x_k))^2 \geq \frac{1}{N} \sum_{x \in \mathcal{X}}\sum_{k=1}^N \I_{\{x_k = x\}} \left(y_k  - \bar{y}_x \right)^2,
\end{equation}
with
\begin{equation}
\label{eq:bar_y_x}
\bar{y}_x=\frac{\sum_{k=1}^N \I_{\{x_k = x\}} y_k}{\sum_{k=1}^N \I_{\{x_k = x\}}},
\end{equation}
and $\I_{\{\cdot\}}$ an indicator function taking value 1, if the event in the brackets is satisfied, and 0, otherwise.
\end{prop}
\begin{proof}
The MSE can be written as
$$
\frac{1}{N} \sum_{k=1}^N (y_k-\mathcal{M}_\theta(x_k))^2 = \frac{1}{N} \sum_{x \in \mathcal{X}}\sum_{k=1}^N \I_{\{x_k = x\}} \left(y_k  - \mathcal{M}_\theta(x) \right)^2.
$$
The right hand side of the above equation can be bounded by calculating lower bounds to each group of summands involving the same regressor value $x$. Given $x$, let $z_x = \mathcal{M}_\theta(x)$, then the corresponding group of summands can be written as
$$
S_x := \frac{1}{N}\sum_{k=1}^N \I_{\{x_k = x\}} \left(y_k  - z_x \right)^2.
$$
The derivative of $S_x$ with respect to $z_x$ is
\begin{align*}
\frac{\partial}{\partial z_x} S_x
&=
-\frac{2}{N}\sum_{k=1}^N \I_{\{x_k = x\}} (y_k-z_x) \\
&= 2 \frac{z_x}{N}\sum_{k=1}^N \I_{\{x_k = x\}} - \frac{2}{N}\sum_{k=1}^N \I_{\{x_k = x\}} y_k.
\end{align*}
Solving the optimality condition $\frac{\partial}{\partial z_x} S_x = 0$ for $z_x$ reveals that the minimum is obtained at \eqref{eq:bar_y_x}. Thus, the lower bound is attained if $\forall x \in \mathcal{X} : \mathcal{M}_\theta(x) = \bar{y}_x$.
\end{proof}

The lower bound given in \eqref{eq:MSELowerBound} is always true, but it is not necessarily a tight bound. If every regressor's value, $x$, appears only once in the data, then this bound would be 0, which is trivial for a sum of squares. The bound will give a nonzero value in the case of observations
with $|\mathcal{X}| < N$, where $|\mathcal{X}|$ denotes the cardinality of the set $\mathcal{X}$.
In our case, when considering $\mathcal{X}=\{3.5,3.6,\ldots,14.9,15\}$, with $|\mathcal{X}| = 116$ and $N=12849$, the bound is non-zero.

\subsubsection{Constrained model}
\label{sec:IncorporatingBounds}

In order to keep the notation and the calculations simple, without loss of generality, we estimate the corresponding  parameters and choose the best WTPC model (within a class), only for wind speeds in the range $[3.5, 15]$. To this end, we consider a slightly modified model: For a model $ \mathcal{M}_{\bm \theta}$ determined by the parameter vector ${\bm\theta}$, from the model class $\mathcal{M}$, we define the constrained model $\overline{\mathcal{M}}_{\bm\theta}$ as
\begin{equation}\label{Constrained}
\begin{split}
\overline{\mathcal{M}}_{\bm\theta}(w) = & \mathcal{M}_{\bm \theta}\big( 3.5 \cdot \I\{w < 3.5\}  + w\cdot \I\{3.5 \leq w < 15\} \\
	& + 15 \cdot \I\{15 \leq w < 25\} \big)  - \mathcal{M}_{\bm \theta} \big( 0 \big)\cdot \I\{25 \leq w \}.
\end{split}
\end{equation}
The argument of $\mathcal{M}_{\bm \theta}$ is constructed such that for wind values smaller than $3.5$ the model $\overline{\mathcal{M}}_{\bm \theta}$ will result in the same power output values as $\mathcal{M}_{\bm \theta}(3.5)$, for values $w\in[3.5,15)$ the model $\overline{\mathcal{M}}_{\bm \theta}$ results in the same power output as $\mathcal{M}_{\bm \theta}$, for values $w\in[15,25)$ the model $\overline{\mathcal{M}}_{\bm \theta}$ results in the same power output as $\mathcal{M}_{\bm \theta}(15)$, and for wind values above $25$ the predicted power output is zero. Considering the constrained model, we can estimate the parameters of the model as usual after a slight transformation of the training data: for all observations with $w< 3.5$ the value of $w$ is changed to $3.5$, for all observations with $w\in[15, 25)$ the value of $w$ is changed to $15$, and all observations with $w \geq 25$ are ignored. Then, this transformed dataset is used for the parameter estimation.

\subsubsection{Logistic models}
\label{sec:Logistic}

Logistic models have been widely used in growth curve analysis and their shape resembles that of a WTPC under the cut-out speed. For this reason, they were recently applied to model WTPCs, see \citep{KusiakZhengSong2009,Lydia2013}. \cite{Lydia2014452} present an overview of parametric and non-parametric models for the modeling of the WTPC, and state that the 5-parameter logistic (5-PL) function is superior in comparison to the other models under consideration. However, as we show in the sequel, this statement should be viewed with skepticism and perhaps should be interpreted as the result of a comparison only within models with the same number of parameters (parametric models) or same level of complexity (non-parametric models).

In this section, we apply a different formulation of the logistic model from the one used in \citep{Lydia2014452}, so as to improve fitness. The 5-PL model used in \citep{Lydia2013} is given as follows
\begin{equation}\label{5pl}
	p = \theta_{5} + \frac{\theta_{1}-\theta_{5}}{\left(1+\left(\frac{w}{\theta_{2}}\right)^{\theta_{3}}\right)^{\theta_{4}}}.
\end{equation}
In this model, parameters $\theta_{1}$ and $\theta_{5}$ are the asymptotic minimum and maximum, respectively, parameter $\theta_{2}$ is the inflection point, parameter $\theta_{3}$ is the slope and $\theta_{4}$ governs the non-symmetrical part of the curve. However, this type of 5-PL does not describe the asymmetry as a function of the curvature, see \citep{RickettsHead1999}.
As an alternative, \cite{Stukel1988} proposed a technique which can handle the curvature in the extreme regions. We apply this technique with a slight modification to fit a logistic model to the WTPC. The general form of the model is presented in Equation \eqref{Stukel}
\begin{figure*}[!t]
\normalsize
\begin{equation}\label{Stukel}
    p= \mathcal{G}_{\bm\theta}(w)=\theta_{1}+\frac{\theta_{4}-\theta_{1}}{1+exp\left[-\left\{ \theta_{2}\left(w-\theta_{3}\right)
    +\theta_{\ell}\left(w-\theta_{3}\right)^{2}\mathds{1}_{_{[3.5,\theta_{3})}}\left(w\right)
    +\theta_{u}\left(w-\theta_{3}\right)^{2}\mathds{1}_{_{\left[\theta_{3},15\right]}}\left(w\right)
      \right\} \right]},
\end{equation}
%
\hrulefill
\vspace*{4pt}
\end{figure*}
with $\mathds{1}_A(w)$ denoting the indicator function taking value 1 when $w\in A$, for a given set $A$, and zero otherwise.
We substitute the term $\left(w-\theta_{3}\right)^{2}\mathds{1}_{_{[3.5,\theta_{3})}}\left(w\right)$ in \eqref{Stukel} with  $\left(w-\theta_{3}\right)^{4}\mathds{1}_{_{[3.5,\theta_{3})}}\left(w\right)$, so as to capture more accurately the curvature in the left tail of the WTPC, cf. Figure \ref{fig:2013_full-clean-overlayCleanWTPC}. For this reason, we refer to this model as the modified Stukel model (mStukel).

In Table \ref{tab:ModelFit}, we present the MSE and BIC for the 5-PL model and the mStukel model according to \eqref{Constrained}. As it is evident from the results presented in Table \ref{tab:ModelFit}, the mStukel model drastically improves the fitness of the WTCP. The parameter estimates, $\hat{\bm{\theta}}^{(g)}$, of the mStukel model with the corresponding standard errors are given in Table~\ref{tab:LogisticParameterEstimates}.

\begin{table}[h!]
    \centering
    \caption{MSE and BIC for the fitted models on the training and validation datasets}
    \label{tab:ModelFit}
     \begin{tabular}{lcccc}\hline
            &\multicolumn{2}{c}{\underline{\bf{Training dataset}}} &&\underline{\bf{Validation dataset}} \\
       \bf{Models}             & \bf{MSE}&\bf{ BIC}   &&   \bf{MSE}\\ \hline
       5-PL                    & 1554.2700  & 131000  &&  1650.7300 \\
       mStukel  &  884.4321  & 123710  && 1020.3800\\  \hline
     \end{tabular}
\end{table}

\begin{table}[h]
\centering
\caption{Parameter estimates (StandardError) for the mStukel logistic model on the training dataset}
\label{tab:LogisticParameterEstimates}
\begin{tabular}{cc}\hline
     \bf{ Parameter} &  \bf{  Training set } \\\hline
      $\hat\theta_1$ &     -30.8580(1.3187) \\
      $\hat\theta_2$ &     ~0.5845(0.0010)   \\
      $\hat\theta_3$ &     ~9.6481(0.0032)  \\
      $\hat\theta_4$ &     ~2010.46(1.2119)  \\
      $\hat\theta_u$ &     ~0.1602(0.0019) \\
      $\hat\theta_\ell$ &  -0.0010(0.00004)\\\hline
\end{tabular}
\end{table}

\subsubsection{Piecewise linear model class}
\label{sec:PiecewiseLinear}

Piecewise linear models are not particularly appealing for practical use for many reasons, but they are very useful as benchmarks. We include piecewise linear models so they can serve as a benchmark non-parametric  model class and because, as it is shown in the sequel, cf. Proposition \ref{prop:PWL_MSE_bound}, this class can achieve the bound of the MSE.

Let the piecewise linear function be defined as follows
\begin{equation}
\label{eq:linearPowerModelNoBoundary}
p =  \mathcal{L}_{\bm\theta^{(\ell)}}(w)= \theta + \sum_{k=0}^{m-1} \I\{s_k \leq w\} (w - s_k)\theta_{k}.
\end{equation}
The parameter vector $\bm\theta^{(\ell)}$ consists of the (height) parameter $\theta$ and the segment slope parameters $\theta_{k}$, $k=0,1,\ldots,m-1$. The splitting points $s_0, \ldots, s_{m-1}$ should be defined beforehand. Throughout the paper,  we use equidistant splitting points on the interval $[3.5, 15]$ and we estimate the parameters of the constrained model $\bar{\mathcal{L}}_{{\bm \theta}^{(\ell)}}$ defined in  Section~\ref{sec:IncorporatingBounds}.

The piecewise linear model can achieve the bound of the MSE on the training data. This is due to the quantized nature of the values of the data to one decimal digit. Thus, using 116 splitting points for the piecewise linear model, we can cover the entire range of wind values in $[3.5,15]$. In this case, the least-squares estimate of the power output is given as the average of the power values of samples given the value of the wind speed, thus attaining the lower bound of the MSE on the training data.

\begin{prop}[Piecewise linear model attaining the lower bound of the MSE]
\label{prop:PWL_MSE_bound}
For a scalar dataset with one dimensional regressors with $|\mathcal{X}|=m+1$ a piecewise linear model of order $m$ with split points $\mathcal{X}$ attains the minimal MSE bound given in Proposition~\ref{stat:MSELowerBound}.
\end{prop}
\begin{proof}
For $|\mathcal{X}|=1$ the only parameter to be estimated is $\theta$, which should be chosen as $\bar{y}$, cf.  Proposition~\ref{stat:MSELowerBound}.

The rest of the proof is based on induction on the cardinality of the set $\mathcal{X}$, denoted by $|\mathcal{X}|$. Let $\left(x^{(i)}\right)_{0\leq i \leq m}$ be the ordered values of $\mathcal{X}$, such that $x^{(0)}<x^{(1)}<\cdots<x^{(m)}$ and lets assume that the parameters $\theta$, $\theta_0,\ldots,\theta_{m-2}$ are chosen such that the linear model with these parameters attains the minimal MSE on the restricted dataset having regressors $\mathcal{X} \setminus \{x^{(m)}\}$. To prove the statement, we need to show that $\theta_{m-1}$ can be chosen such that $\mathcal{L}_{\bm \theta}(x^{(m)}) = \bar{y}_{x^{(m)}}$. From the definition of the piecewise linear function
\[
\mathcal{L}_{{\bm \theta}^{(\ell)}}(x^{(m)}) = \theta + \sum_{k=0}^{m-2}(x^{(m)} - x^{(k)})\theta_k + (x^{(m)} - x^{(m-1)})\theta_{m-1}.
\]
Solving this equation for $\theta_{m-1}$, we get that
\[
\theta_{m-1}= \frac{\bar{y}_{x^{(m)}} - \theta - \sum_{k=0}^{m-2}(x^{(m)} - x^{(k)})\theta_k}{x^{(m)} - x^{(m-1)}},
\]
which concludes the proof.
\end{proof}

It can be stated in general that once a model structure has enough degrees of freedom to assign the estimates $\mathcal{M}_{\bm\theta}(w)$ independently to every wind value $w\in \mathcal{X}$, then the lower bound for the MSE value can be attained.

Piecewise linear model classes, with a fixed number of splitting points equidistantly chosen in an interval, are not properly nested, if $m=1,2,\ldots$. Proper nesting is achieved, if $m=2^0,2^1, 2^2, \ldots$, or if some other exponential series is chosen. The parameters $\hat{\bm \theta}_m^{(\ell)}$ of a model belonging to the fixed choice of $m$ can be estimated for different values of $m$, but then the problem reduces to optimally choosing $m$, which is a model selection problem.

The other reason, why it is instructive to examine the properties of the piecewise linear model structure, is that, assuming Gaussian residuals, the combined variance of the estimated parameters can be calculated analytically. This is visualized in Figure~\ref{fig:linearModelCovarianceTrace} as the trace of the estimated covariance matrix of the parameters is shown against the complexity of the model class $m$. This shows the generic features of model selection problems.

For very small values of $m$ the modeling error is big, so the estimated variance of those few estimated parameters is going to be big (combination of modeling error and variance from noise), so the sum is going to be a sum of few but large in absolute value entries. Values of $m$ that correspond to a model class that can properly model the data will result in a sum that contains more summand terms, but with smaller in absolute value entries. The variance of the parameters in this case is expected to be small for two reasons: \emph{i)} the modeling error is reduced or eliminated; \emph{ii)} a small number of parameters needs to be estimated from the data. For higher values of $m$ the number of summands will increase and so will the corresponding absolute values of the entries. This is because the modeling error was already minimized and a higher number of parameters needs to be estimated from the data, which increases their variance. This heuristic results in a quasiconvex shape of the MSE as a function of the complexity parameter $m$ (a.k.a. the model order). The goal of model selection is to define how an optimal model order $\hat{m}$ should be chosen. This question arises in the case of all non-parametric  model classes and our approach is based on the BIC, see Section~\ref{sec:ModelSelection}.

As it can be seen in Figure~\ref{fig:linearModelCovarianceTrace}, the trace of the estimated covariance matrix of the parameters is quasiconvex shaped, as it decreases in the beginning and then it increases rapidly when the model order is increased. If this is compared to the decrease of the MSE shown in Figure~\ref{fig:ModelComparisonTraining}, we see that going above a given complexity level just adds unnecessary uncertainty to the estimation without improving the modeling precision. This trade-off should be balanced by the model selection algorithm. Using model selection based on the BIC, see Section~\ref{sec:ModelSelection}, the optimal number of linear segments turns out to be $m = 13$. The parameters of the estimated model are given in Table~\ref{tab:LinearEstimate}. The MSE of this model on the training data is $815.5127$, while on the validation set it is $974.6084$.

\begin{figure}
	\centering
	\includegraphics{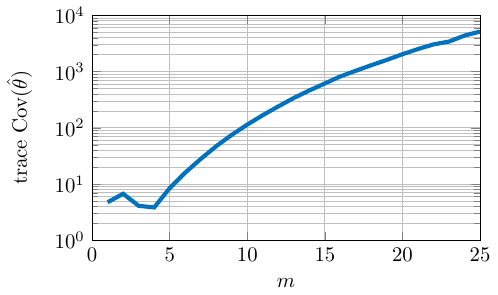}
	\caption{The sum of the parameter variances for the piecewise linear model assuming Gaussian residuals}
	\label{fig:linearModelCovarianceTrace}
\end{figure}

\begin{table}
\caption{Estimated parameters of a piecewise linear model with $m=13$ segments}
\label{tab:LinearEstimate}
\begin{center}
\begin{tabular}{cc|cc}
\hline
\multicolumn{4}{c}{$t_k = 3.5 + k\frac{15-3.5}{13}$, $k=0, \ldots,12$ } \\
\hline
 \bf{ Parameter} &  \bf{  Training set }& \bf{ Parameter} &  \bf{  Training set } \\\hline
$\hat{\theta} $&-8.3398  & $\hat{\theta}_0$&0.8929 \\
$\hat{\theta}_1$& 95.4169  & $\hat{\theta}_2$& 17.8948\\
$\hat{\theta}_3$& 42.1545  & $\hat{\theta}_4$& 58.0053\\
$\hat{\theta}_5$&  38.8957 & $\hat{\theta}_6$& 60.8006\\
$\hat{\theta}_7$& 47.2667  & $\hat{\theta}_8$&6.8485 \\
$\hat{\theta}_9$& -68.7682 & $\hat{\theta}_{10}$& -207.2188\\
$\hat{\theta}_{11}$& -94.4818 & $\hat{\theta}_{12}$& 4.8853\\
\hline
\end{tabular}
\end{center}
\end{table}

\subsubsection{Polynomial model class}
\label{sec:Polynomial}

A univariate polynomial model of degree $m$ of the power function is given as
\begin{equation}
\label{eq:polyPowerModelNoBoundaryUnstable}
p=\mathcal{P}_{\bm\theta}(w)= \sum_{i=0}^m \theta_i w^i.
\end{equation}
The formulation given in \eqref{eq:polyPowerModelNoBoundaryUnstable} should be adapted to take into account the constrained model $\overline{\mathcal{P}}$ defined in Section~\ref{sec:IncorporatingBounds}. However, even after the transformation to the constrained model, and the reduction of the wind range to practically $[3.5, 15]$, we have to note that estimating the parameters $\theta_i$, $i=0,\ldots, m$, of the polynomial model is a numerically difficult problem. This is due to the fact that, e.g., for a polynomial model of degree $m=14$, the coefficient matrix of the parameters includes entries corresponding to values $1,15,15^2,\ldots,15^{14}$. Inverting such a matrix is numerically unstable due its high condition number, cf. \citep{belsley2005regression}.

To overcome the numerical stability issues, one of the simplest techniques is to rescale the argument $w$ of the polynomial, so higher powers of the argument will still remain numerically tractable. With this change, we redefine the polynomial model as
\begin{equation}
\label{eq:polyPowerModelNoBoundary}
\mathcal{P}_{{\bm\theta}^{(p)}}(w) = \bar{p} + d_p  \sum_{i=0}^m  \theta_i \left(\frac{w-\bar{w}}{d_w}\right)^i ,
\end{equation}
where the polynomial parameter vector $\bm\theta^{(p)}$ contains the coefficients of the polynomial $\mathcal{P}$ as well as the scaling parameters $\bar{w}$, $\bar{p}$, $d_w$, $d_p$. $\bar{w}$ and $\bar{p}$ denote the sample averages of the wind speed ($w$), and the sample average of the power output ($p$), respectively, while $d_w$ and $d_p$ denote the sample standard deviation of the wind speed ($w$), and of the power output ($p$), respectively. The model order parameter for polynomial models is the degree of the polynomial, $m$.

Polynomial models are not performing well according to the literature. This is the result of a combination of factors: Firstly, they are not capable of capturing the flat plateau on the left and the right tail of WTPC. Once this obvious drawback is compensated by considering the constrained model, polynomial models drastically increase their fitness. Secondly, there are various numerical difficulties associated with the estimation of the parameters of polynomial models. Unfortunately, this issue constitutes a significant drawback especially at higher model orders, as we show in Section~\ref{sec:ComparisonOfModelClasses}.

Estimating (in the least squares sense) the coefficients of a polynomial with degree $m=14$ results in the parameters presented in Table~\ref{tab:PolynomialEstimate}. The choice of degree $m=14$ is explained in Section~\ref{sec:ModelSelection}. The MSE of this model on the training data is $812.2287$, while on the validation set it is $969.8870$.

Note, that the polynomial coefficients in Table~\ref{tab:PolynomialEstimate} are reported with 15 decimal digits, as we take into account the support of the wind values $[3.5,15]$ and the maximum degree of the polynomial model: this is due to the fact that for example for  the polynomial of degree $m=14$ the leading coefficient of the polynomial, $\hat{\theta}_{14}$, is multiplied with $ \left(\frac{w-7.241154953692900}{3.092009009051451}\right)^{14}$ for $w\in[3.5,15]$. This illustrates that the estimation of the polynomial coefficients is numerically sensitive, which is not the case for the other discussed non-parametric model classes.

\begin{table}[h!]
\caption{Estimated parameters of a polynomial model with degree $m=14$}
\label{tab:PolynomialEstimate}
\begin{center}
\begin{tabular}{cc|cc}
\hline
 \bf{Par.} &  \bf{Training set }& \bf{Par.} &  \bf{  Training set } \\
\hline
$\hat{\bar{p}}$&$1012.7$ & $\hat{\bar{w}}$&$7.241154953692900$ \\
$\hat{d}_p$&$601.0210490157367$ & $\hat{d}_w$&$3.092009009051451$ \\
$\hat{\theta}_0$&$-1.083983804472287$ & $\hat{\theta}_8$&$0.913785265115761$ \\
$\hat{\theta}_1$&$1.027493542215327$ &  $\hat{\theta}_9$&$0.158488326138962$  \\
$\hat{\theta}_2$&$0.437620131289974$ &  $\hat{\theta}_{10}$&$ -0.462562253267288$    \\
$\hat{\theta}_3$&$-0.258311524269187$ & $\hat{\theta}_{11}$&$0.100868720012586$ \\
$\hat{\theta}_4$&$0.152839963020718$ & $\hat{\theta}_{12}$&$0.068782606353010$\\
$\hat{\theta}_5$&$0.837258874326937$ & $\hat{\theta}_{13}$&$-0.034900832592028$\\
$\hat{\theta}_6$&$-0.693269004241413$ & $\hat{\theta}_{14}$&$0.004495461408312$\\
$\hat{\theta}_7$&$-0.791866808461089$ & \\ \hline
\end{tabular}
\end{center}
\end{table}

\subsubsection{Spline model class}
\label{sec:Spline}
Splines provide a universal family for approximating smooth functions. A spline is defined by a series of knot points and by polynomials representing its value between the knot points in a continuous way \citep{schumaker2007spline}. Formally a cubic B-spline is given by a triplet of parameters ${\bm{\theta}^{(s)}=(m, \bm{k}, \bm{\alpha})}$, where $m \in \mathbb{N}_+$ is the number of basis functions used,  $\bm{k} \in \R^{m+4}$ is a vector of knot points in nondecreasing order, $\bm{\alpha} \in \R^m $ is the vector of coefficients for the basis functions $B_{i, \bm{k}, 3}$ defined by the Cox-de Boor recursion \citep{Boor1978}
\begin{equation}
\label{eq:CoxdeBoor}
\begin{aligned}
B_{i, \bm{k},0} (x) &= \I_{[{k}_i ,{k}_{i+1})}(x), \ i=1, \ldots , m+3,\\
B_{i, \bm{k}, d} (x)&= \frac{x-{k}_i}{{k}_{i+d}-{k}_i}B_{i, {k},d-1} \\
&+ \frac{{k}_{i+d+1}-x}{{k}_{i+d+1}-{k}_{i+1}}B_{i+1, {k},d-1}, \ d=1,2,3,\\
& i=1, \ldots, m+3-d.
\end{aligned}
\end{equation}
A cubic B-spline model for the WTPC is given as
\begin{equation}
\label{eq:SplinePowerModelNoBoundary}
p = \mathcal{S}_{\theta^{(s)}}(w) = \sum_{i=1}^m \bm{\alpha}_i B_{i, \bm{k},3} (w).
\end{equation}
The complexity of cubic spline models is defined by the number of basis functions $m$. If the knot points $\bm k$ are fixed then the parameters $\bm \alpha$ can be estimated analytically in the least-squares sense, but this cannot be done simultaneously with the location of the knots \citep{Kang2015179}. We use a simple suboptimal procedure to find the estimates, which performs the estimation in two rounds. In the first round the knot points are equidistantly chosen in the $[3.5, 15]$ interval and the parameters $\bm \alpha$ are estimated. In the second round, new knot points are calculated, based on the data and the first round estimates, using the MATLAB\textsuperscript{\textregistered} routine \emph{newknt}, which reallocates the knot points to allow a better estimation of $\bm \alpha$. Then $\bm \alpha$ is estimated for the second time.

The estimated parameters $\hat{\bm{\theta}}^{(s)}$ of a cubic spline using $m=17$ B-splines are
\begin{equation}
\label{eq:parBSpline1}
\begin{split}
\hat{\bm{\alpha}} = [& -8.0336698 \
                       -7.2559215 \
                       -23.865741 \\ &
                       78.529492 \
                       156.55003 \
                       272.98557 \\ &
                       452.80144 \
                       690.69908 \
                       1022.923  \
                       1400.7208 \\ &
                       1721.1444 \
                       1921.2212 \
                       1998.4378 \
                       1992.549  \\ &
                       2005.308 \
                       1997.8069 \
                       2000.3969]
\end{split}
\end{equation}
with knot points
\begin{equation}
\label{eq:parBSpline2}
\begin{split}
\hat{\bm{k}} =  [&3.5 \
                 3.5 \
                 3.5 \
                 3.5 \
                 4.4247 \
                 5.2668 \
                 5.9855 \\ &
                 6.7569 \
                 7.6994 \
                 8.6481 \
                 9.7265 \\ &
                 10.8994 \
                 11.6831 \
                 12.3575 \
                 12.9990 \\ & 
                 13.6470 \
                 14.3235
                 15 \
                 15 \
                 15 \
                 15].
\end{split}
\end{equation}

An interesting feature of the resulting $\hat{\bm{k}}$ is that its first four entries and last four entries coincide. As it can be deduced from Equation \eqref{eq:CoxdeBoor}, the multiplicity of the knot points shows how smooth is the function at the specific knot point. The two endpoints (due to their high multiplicity) indicate that the higher order derivatives are zero at the endpoints of the support $[3.5, 15]$, so the estimated WTPC is flat at the left and right tails of the support. This is expected and desired, since the support was chosen so that the WTPC outside this support is constant (left and right tail of the WTPC).

B-splines are zero outside the range defined by the knot points, so a proper power function estimate is obtained by transforming $\mathcal{S}$ to the constrained model $\overline{\mathcal{S}}$ defined in Section~\ref{sec:IncorporatingBounds}. The MSE of the model $\overline{\mathcal{S}}$, with the parameters given above, evaluated on the training data, is $811.9171$, while on the validation set, it is $969.5854$.

We note that \cite{ShokrzadehJozani2014} developed a much more evolved procedure for the selection of  the number of the knot points, as well as for the selection of the location of the knot points, but such a complicated model choice does not improve more than 1\% the modeling fit, which is insignificant if compared to the improvements achieved by incorporating the wind direction and the ambient temperature, and by the addition of the dynamic layer. 

\subsubsection{Model selection based on BIC}
\label{sec:ModelSelection}

In the case of non-parametric models, the model class consists of subclasses indexed by the complexity parameter $m$ (a.k.a. model order), i.e., piecewise linear models with an increasing number of segments, polynomial models with an increasing degree, or spline models with an increasing number of basis functions. The appropriate model order should be selected in a way that adheres to the principle of parsimony: Goodness-of-fit must be balanced against model complexity in order to avoid overfitting--that is, to avoid building models that in addition to  explain the data, they also explain the independent random noise in the data at hand, and, as a result, fail in out-of-sample predictions.

There are several approaches for selecting a model, among others the AIC \citep{AkaikeIC} or the BIC \citep{SchwarzBIC}. Although AIC can be asymptotically optimal under certain conditions, BIC penalizes the model complexity stronger. Therefore we use the BIC for the selection of the models reported in the previous sections. The BIC is defined as
\begin{equation*}
\mathrm{BIC} (\hat{\bm \theta}_m) = \ln(N) n_{\hat{\bm \theta}} - 2 \ln (\hat L),
\end{equation*}
where $N$ is the number of data samples used to estimate $\hat{\bm \theta}_m$, $n_{\hat{\bm \theta}}$ is the number of estimated parameters and $\hat L$ is the estimated likelihood of the observations assuming the model with estimated parameters $\hat {\bm \theta}_m$.

Assuming a Gaussian noise model $p_k = \mathcal{M}_{\tilde{\bm \theta}}(w_k) + \varepsilon_k$, $k=1,\ldots,N$, where $\varepsilon_k \stackrel{\text{i.i.d.}}{\sim} \mathcal{N}(0, \sigma^2)$, we get that the BIC can be written as
\begin{equation*}
\mathrm{BIC} (\tilde{{\bm \theta}}) = \ln(N) n_{\tilde{{\bm \theta}}} + N \ln(2 \pi \sigma^2) + \frac{\sum\limits_{i=1}^N \varepsilon_i^2}{\sigma^2}.
\end{equation*}

If the parameters $\hat{\bm \theta}_m$ of a model from the subclass with complexity $m$ are estimated using the training data, then the MSE on the training data, say $\mathrm{MSE}_{m}$, is an asymptotically unbiased estimator for the unknown variance $\sigma^2$. Thus, evaluating the BIC on the training data yields that
\begin{equation*}
\mathrm{BIC} (\hat{\bm \theta}_m) \approx \ln(N) n_{\hat{\bm \theta}_m} + N \ln( \mathrm{MSE}_{m}) + N \ln(2 \pi) + 1.
\end{equation*}

Models with different complexity are compared using the $\mathrm{BIC} (\hat{\bm \theta}_m)$ and the model complexity is estimated as
\begin{equation*}
\hat{m} = \argmin_m  \mathrm{BIC}  (\hat{\bm \theta}_m),
\end{equation*}
resulting in the final estimate
\begin{equation*}
\hat{\bm \theta} = \hat{\bm \theta}_{\hat m}.
\end{equation*}

Using this procedure, we obtain that for the piecewise linear models, the optimal number of segments is $13$, for the cubic spline models the optimal number of basis functions is $17$, while for the polynomial models the optimal degree is $14$.

\subsection{Comparison of models from different classes}
\label{sec:ComparisonOfModelClasses}

Figures~\ref{fig:ModelComparisonTraining}~and~\ref{fig:ModelComparisonValidation} depict the behavior of the MSE as a function of the complexity for the different model structures on the training and the validation datasets, respectively. The goal of this section is to summarize the remarks that can be made based on these figures. Since, the mStukel logistic model has a fixed number of complexity parameters $m$ (a.k.a. fixed model order), its MSE is depicted in Figures~\ref{fig:ModelComparisonTraining}~and~\ref{fig:ModelComparisonValidation} as a constant function in $m$, taking values 884.4321 and 1020.3800, on the training and validation sets, respectively.

\begin{figure}[!h]
     \begin{subfigure}[b]{0.9\linewidth}
            	\includegraphics[scale=0.9]{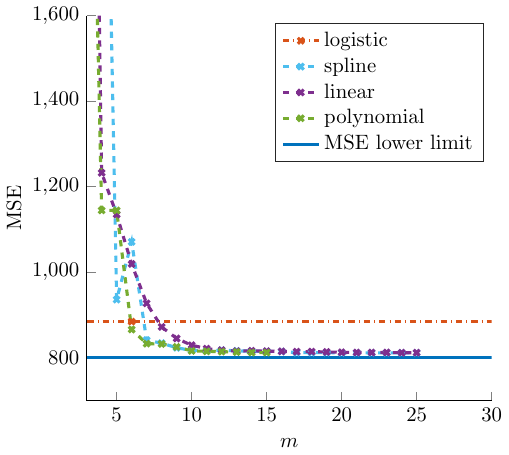}
	\caption{The MSE on the training set}
	\label{fig:ModelComparisonTraining}
	\end{subfigure}\\
\begin{subfigure}[b]{0.9\linewidth}
           	\includegraphics[scale=0.9]{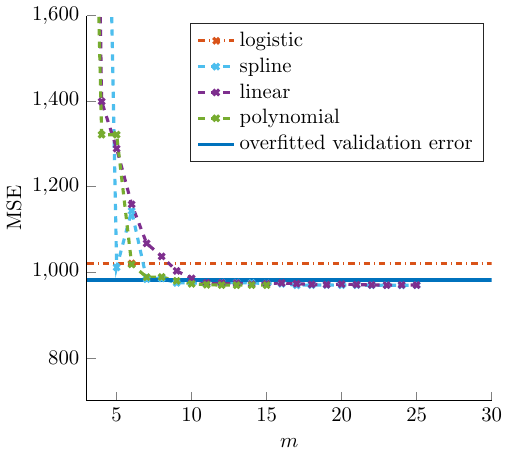}
	\caption{The MSE on the validation set}
	\label{fig:ModelComparisonValidation}
    \end{subfigure}
    \caption{The MSE of different model types on the training set (\protect\subref{fig:ModelComparisonTraining}) and on the validation set  (\protect\subref{fig:ModelComparisonValidation}) as a function of model complexity}
\end{figure}

The logistic model class, due to its parametric nature, only contains models that have a specific shape similar to what is expected from the WTPC. This is the reason why it performs relatively well, if compared to other models with matching complexity $(m=6)$, especially when compared to the piecewise linear and the spline models. The main advantage of parametric models is that they have a pre-defined shape that matches the data, and they can describe the model with a much smaller number of parameters. As a result they can be used in case the data sparsely covers the support. However, in our case, due to the large amount of data covering densely the full support of the wind values, such advantage does not become apparent.

As stated in Proposition~\ref{stat:MSELowerBound}, we can calculate the MSE lower bound based on the data. Regardless of the model class, the MSE converges to that limit, as the complexity parameter tends to infinity, $m \rightarrow \infty$. Moreover, for some of the model classes we  investigate, convergence will occur with finite complexity, e.g., piecewise linear models converge at $m=116$, cf. Section~\ref{sec:PiecewiseLinear}. Similar complexity values can be calculated for the other model classes. It is important to note that the MSE converges rapidly to the lower bound for small values of $m$, while for large values of $m$, convergence seems to slow down significantly. As it can be seen in Figure~\ref{fig:ModelComparisonTraining}, this happens in the range $m \in [10,15]$, depending on the model class.

As expected, when considering a very complicated model, then the validation error has the tendency to increase in comparison to the optimal complexity model. The solid line in Figure~\ref{fig:ModelComparisonValidation} depicts the validation error of the model that attains the lower limit of the estimation error. As it is visible in the figure, this overfitting error is really small in comparison to the validation error of model orders around the optimal order (the validation error of the most overfitted model is 955.9658 that is approximately 1\% worse than the validation errors of the different models). This shows that the data are fully covering the support $[3.5, 15]$ of the wind speed for the constrained model, and that, in our case, overfitting issues are of minor importance, as such overfitting does not impact significantly the validation error.

The model orders selected by the BIC, cf. Section~\ref{sec:ModelSelection}, are all under $m=20$, so they result in relatively simple models. They require more parameters than the logistic model, but the comparison based on the BIC indicates that the extra flexibility of these models is needed. This is not unexpected, given the provided improvement in terms of the modeling error.

Here we can underline one of the main messages of the paper: non-parametric models seem to be more suitable for the WTPC modeling than parametric models. This is mainly due to the relatively simple shape of the WTPC and the large amount of available data that can be used for the estimation of models with high complexity.

Another advantage for non-parametric models is that their shape does not depend significantly on the knowledge of the cut-in speed and the rated speed. As long as a lower bound for the cut-in speed and an upper bound for the rated speed are known, these models obtain the right shape between the two values. Whereas, for parametric models these two values are either chosen beforehand, which impacts the model significantly, or they need to be estimated, which impacts significantly the difficulty of estimation (e.g., non-convexity, non-uniqueness, etc.).

Based on the above remark, that non-parametric models provide a better fit for the WTPC, we now turn our attention to the natural question on how to choose between the various classes of non-parametric models. This question is treated in the sequel in more detail.

In what follows, the goal is to show that in theory it should not matter which non-parametric model class we choose to estimate the WTPC, however practical considerations can still result in arguments against particular model classes. The main objective, when considering a model class, should be the numerical robustness of the estimation procedure that can provide the corresponding estimates.

The polynomial model structure is evaluated only up to degree $m=15$, cf. Figures~\ref{fig:ModelComparisonTraining} and \ref{fig:ModelComparisonValidation}. This is because estimating higher order polynomials is numerically infeasible, as we already mentioned in Section~\ref{sec:Polynomial}. When it comes to estimation of splines with fixed knot points $\bm k$, the estimates of the coefficients $\bm \alpha$ can be obtained in a numerically reliable way. Similarly for piecewise linear models, given the split points $(t_{k})$, the estimation is numerically reliable. Due to the simple shape of the WTPC, the allocation of these points is not particularly important. What could be gained by the optimal choice of these points is on the one hand negligible, as it can be seen from Figure~\ref{fig:ModelComparisonTraining}, and on the other hand it can be mitigated by adding extra parameters.

In order to illustrate that the choice of the model class is almost irrelevant, we define a measure of comparison for models from different model classes, say $\mathcal{M}_{\bm \theta_i}(\cdot)$, $i=1,2$, as follows
\begin{equation}
\Delta_{\bm \theta_1, \bm \theta_2} = \frac{\mathds{E}(\mathcal{M}_{\bm \theta_1}(W)-\mathcal{M}_{\bm \theta_2}(W))^2}{\min\left\{\mathds{E}(P-\mathcal{M}_{\bm \theta_1}(W))^2, \mathds{E}(P-\mathcal{M}_{\bm \theta_2}(W))^2\right\}},
\end{equation}
with $W$ and $P$ denoting the random wind speed and the random power output, respectively.
$\Delta_{\bm \theta_1, \bm \theta_2}$ measures what is the expected difference between predictions made by the two models $\bm \theta_1$ and $\bm \theta_2$ relative to the modeling error of the better of the two. This quantity is evaluated empirically both on the training and on the validation data using the models estimated earlier and the obtained values are reported in Table~\ref{tab:RelativeModelDiff}.
The difference between the logistic and the selected spline model $\Delta_{\hat{\bm \theta}^{(g)}, \hat{\bm \theta}^{(s)}}$ is approximately $10\%$, the difference between the piecewise linear and spline models $\Delta_{\hat{\bm \theta}^{(\ell)}, \hat{\bm \theta}^{(s)}} $ is under $1\%$, and $\Delta_{\hat{\bm \theta}^{(p)}, \hat{\bm \theta}^{(s)}} $ gets even smaller when it comes to the polynomial and spline models.
This indicates that optimizing the model selection with regard to the class is not expected to provide significant improvements. As there is no significant difference between the various non-parametric WTPC model classes, from this point onward, we restrict our analysis to the spline model given in Section~\ref{sec:Spline}.

\begin{table}
\caption{The relative difference between models of different classes}
\label{tab:RelativeModelDiff}
\begin{center}
\begin{tabular}{ccc}
\hline
 \bf{ Parameter} &  \bf{  Training set }& \bf{ Validation set} \\
\hline
$\Delta_{\hat{\bm \theta}^{(g)}, \hat{\bm \theta}^{(s)}} $& 0.0884 & 0.0715 \\
$\Delta_{\hat{\bm \theta}^{(\ell)}, \hat{\bm \theta}^{(s)}} $ & 0.0059 & 0.0048 \\
$\Delta_{\hat{\bm \theta}^{(p)}, \hat{\bm \theta}^{(s)}} $ & 0.0005 & 0.0004 \\
\hline
\end{tabular}
\end{center}
\end{table}

In the next sections, we address two points of concern: \emph{i)} we discuss how to improve the WTPC model by incorporating more environmental variables, such as the relative wind angle and the ambient temperature, and \emph{ii)} we explore if the residuals of the model are Gaussian and investigate how to incorporate the natural autocorrelation of the data into the model by specifying that the power output variable depends linearly on its own previous values and on a stochastic term. In Section~\ref{sec:MorePhysicalParamaters}, we discuss the results of incorporating more environmental variables into the power estimation, while in Section~\ref{sec:TimeSeriesPrediction}, we explore the possibility of estimating the power output based on previous measurements in time.

\section{Including more physical parameters}
\label{sec:MorePhysicalParamaters}

From a physical perspective the power output can be model as
\begin{equation}
\label{eq:physicalmodel}
p= \frac{1}{2}\rho\pi R^2C_p(\lambda,\beta) w^3,
\end{equation}
with $p$ the power captured by the rotor of a wind turbine, $\rho$ the air density, $R$ the radius of the rotor determining its swept area, $C_p$ the power coefficient which is a function of the blade-pitch angle $\beta$ and the tip-speed ratio $\lambda$, and $w$ the wind speed, see \citep[Eq.~(2)]{Lee2015}. Thus, although wind speed is the most relevant factor determining the power output, it is evident from \eqref{eq:physicalmodel} that other environmental or turbine specific factors impact the power output. One way to improve the modeling and forecasting capabilities of the WTPC model is to incorporate additional relevant parameters according to the physical first principle arguments. In accordance to our available data, we illustrate the additional benefits of  incorporating two additional environmental parameters: the relative incidence angle of the wind with respect to the  rotor plane, say $\phi$, and the ambient temperature recorded on the exterior of the wind turbine nacelle, say $T$. 

The new signals are incorporated into \eqref{eq:SplinePowerModelNoBoundary} as follows
\begin{equation}
\label{eq:AdditionalPhysicalLayers}
p = \mathcal{F}_{\bm{\theta}^{(f)}}(w, \phi, T) = \mathcal{S}_{\bm{\theta}^{(s)}}(w \cdot |\cos(\phi)|^{c_\phi})\left(1+c_T \left(T - \bar{T} \right) \right),
\end{equation}
with $\phi$ and $T$  the relative incidence angle and the ambient temperature, $\bar{T}$ the average temperature obtained by the training data, and $c_\phi \geq 0$ and $c_T$ are parameters to be estimated from the data.

The inclusion of these factors can be argued based on heuristic arguments as follows: As a rough approximation, it can be stated that power generation is only achieved by the perpendicular component of the wind speed to the rotor plane of the turbine. This perpendicular component is mathematically  represented by $w \cdot \cos(\phi)$. The introduction of the absolute value of the $\cos$ term ensures that the direction of the wind is not changed. Furthermore, the inclusion of the $c_\phi \geq 0$ parameter in the $|\cos(\phi)|^{c_\phi}$ term makes sure that the wind is not amplified (i.e., the wind speed cannot get a multiplier greater than one).
Regarding the inclusion of the temperature factor, this is motivated by the inherent physical relation of the temperature and the air density, as well as the prominent role of air density in the physical expression of the power output, cf. \eqref{eq:physicalmodel}. Without assuming any specific functional form for this dependence, the parameter $c_T$ can be thought of as the partial derivative of this relationship around the average temperature $\bar{T}$.

Using the B-spline WTPC model $\hat{\bm \theta}^{(s)}$ with complexity $m=17$, with estimated parameters $\hat{ \bm k}$ and $\hat{\bm \alpha}$ given in Section~\ref{sec:Spline}, we can estimate the value of $c_\phi$ and $c_T$ in the least squares sense. In Table~\ref{tab:EnvironmentalParameterEstimates}, we provide the MSE values of the model \eqref{eq:AdditionalPhysicalLayers}, where the parameters are estimated or fixed to zero in different combinations. Fixing either $c_\phi$ or $c_T$ is equivalent to omitting the corresponding modeling aspect. This allows us to see the impact of the different environmental parameters on the power generation.
 The first line contains the baseline, the MSE value of the WTPC model with only the wind factor. The other lines contain the MSE values corresponding to WTPC models generalized to include only the incidence angle; only the relative temperature; or both the angle and the temperature.  In the remainder of the paper, $\hat{\bm{\theta}}^{(f)}$ denotes the combination of the B-spline WTPC model $\hat{\bm{\theta}}^{(s)}$ generalized to include the two environmental factors, with estimates $\hat{c}_\phi=1.0115$ and $\hat{c}_T=-0.0050$. For this reason, when needed, we shall illustrate this by writing $\hat{\bm{\theta}}^{(f)}=(\hat{\bm{\theta}}^{(s)},\hat{c}_\phi,\hat{c}_T)$.

\begin{table}
\begin{center}
\caption{The MSE values of models given in \eqref{eq:AdditionalPhysicalLayers} with different modeling complexity}
\label{tab:EnvironmentalParameterEstimates}
\begin{tabular}{cccc}
\hline
 \bf{ $c_\phi$ parameter} & \bf{ $c_T$ parameter} &  \bf{MSE on }& \bf{ MSE on  } \\
 &   &  \bf{  training set}& \bf{  validation set} \\
\hline
$c_\phi=0$ & $c_T=0$ & 811.9171 & 969.5854\\
$\hat{c}_\phi=0.4279$ & $c_T=0$ & 810.1959 & 943.8046\\
$c_\phi=0$ & $\hat{c}_T=-0.0047$ & 690.5758 & 798.9763\\
$\hat{c}_\phi=1.0115$ & $\hat{c}_T=-0.0050$ & 681.6206 & 753.7712\\
\hline
\end{tabular}
\end{center}
\end{table}

From Table~\ref{tab:EnvironmentalParameterEstimates}, it is evident that the inclusion of the incidence angle does not improve significantly the model. This is evident by the difference between the MSE values of the two models captured in the first and second line, respectively, of the table, or similarly in the third and fourth line.
This difference is smaller than one percent. This is not because the incidence angle is not relevant for power generation,  but because the automatism in the turbine keeps the nacelle facing the most beneficial direction with regard to the power output. Figure~\ref{fig:RelativeAngleDensity} depicts the empirical density function of the incidence angle $\phi$ and it shows that the incidence angle is tightly concentrated around 0 degrees, which indicates that the wind is almost always nearly perpendicular to the rotor plane. Contrary, the inclusion of the temperature, with estimated parameter $c_T$, adds more than $15\%$ in the modeling precision. Figure~\ref{fig:TemperatureDiagram} depicts the values of the ambient temperature for the training set (red) and the validation set (blue). Moreover, the negative sign of $\hat{c}_T$ matches the physical insight that increasing the temperature leads to a decrease of the air density at constant pressure.

\begin{figure}[h!]
\centering
\begin{minipage}{.45\textwidth}
  \centering
\includegraphics{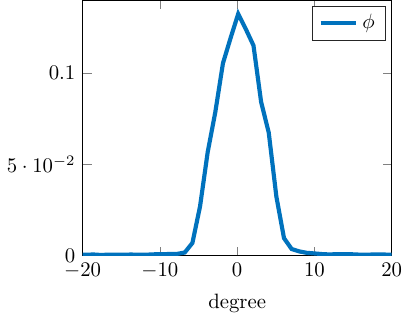}
	\caption{Empirical density of the incidence angle in degrees }
	\label{fig:RelativeAngleDensity}
\end{minipage}\ ~~~\ %
\begin{minipage}{.45\textwidth}
  \centering
	\includegraphics{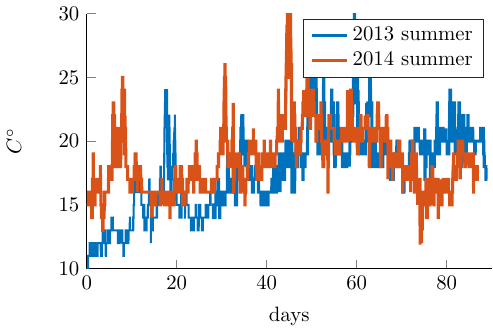}
	\caption{ The temperature variation during the training and the validation period }
	\label{fig:TemperatureDiagram}
\end{minipage}
\end{figure}

All in all, the impact of the temperature is much larger than that of the incidence angle because the turbine mechanisms cannot influence the temperature, as they can the incidence angle. According to \eqref{eq:physicalmodel}, besides the wind speed and the air density (in the form of temperature in our case, due to temperature data availability) there seem to be no other important environmental factors, that may affect the power generation  significantly and that can be predicted well. In the next section we investigate the residuals of the model with the environmental factors and model the power output by adding a time series layer that allows for a significant short-term forecasting improvement.

\section{Predicting power output using time series}
\label{sec:TimeSeriesPrediction}

The effect of the modeling error can be compensated, to some extent, by modeling the residual power output with a stochastic process that has a time scale, which is much slower than that of the wind turbine mechanics. This is based on the intuition that the power output reacts quickly to environmental changes (time scale of minutes), but the environment changes in a slower rate (if there is strong wind, that will last for a few hours with a high probability).

In order to model the power output more accurately, we need to understand the statistical properties of the residuals, which are obtained as follows
\begin{equation}
r = p-\mathcal{F}_{\hat{\bm{\theta}}^{(f)}}(w, \phi, T) = p-\hat p,
\end{equation}
where $\hat p$ is a brief notation for $\mathcal{F}_{\hat{\bm{\theta}}^{(f)}}(w, \phi, T)$.

The standard deviation of the distribution of the residual, conditioned on the wind speed value, is shown in Figure~\ref{fig:ResidualVariances}. It is apparent from the figure, that the standard deviation of the residuals is relatively small below the cut-in speed and above the rated speed, but this is not the case between these values. Thus, the WTPC model (even enhanced with the environmental factors) does not explain fully the power output leaving only the inherent randomness (captured by the residuals). Therefore, we need to further enhance the WTPC model. To this end, we concentrate in the range of wind values for which the standard deviation of the residuals seems to be large.

\begin{figure}[h!]
	\centering
	\includegraphics{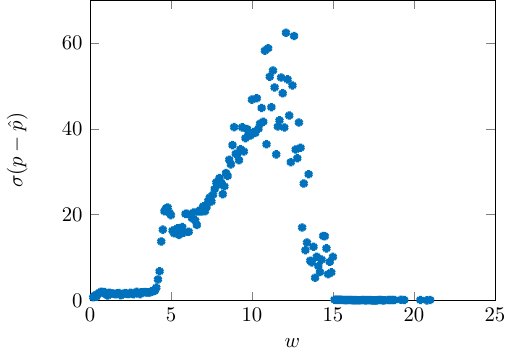}
	\caption{The standard deviation of the residuals $p-\hat{p}$ with respect to the wind speed}
	\label{fig:ResidualVariances}
\end{figure}

Let $\sigma_w$ denote the conditional standard deviation of the residuals $r$ restricted to observations with wind values equal to $w$
\begin{equation*}
\sigma_w = \frac{\left(\sum_{k=1}^N \I_{\{w_k = w\}}\left(p_k-\mathcal{F}_{\hat{\bm{\theta}}^{(f)}}(w_k, \phi_k, T_k)\right)^2\right)^{1/2}}{\left(\sum\limits_{k=1}^N \I_{\{w_k = w\}}\right)^{1/2}}.
\end{equation*}
The rescaled residual signal $r'$ is defined as
\begin{equation*}
r' = \frac{r}{\sigma_{w}}.
\end{equation*}

\begin{figure}[h!]
    \centering
    \begin{subfigure}[b]{0.45\textwidth}
	\includegraphics{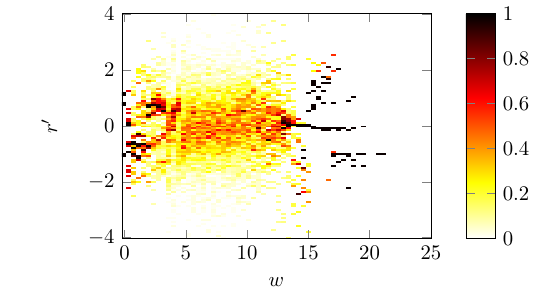}
	\caption{Wind support $[0,25]$}
	\label{fig:RescaledResidualDistributionFH}
    \end{subfigure}
    ~~~~~~
    \begin{subfigure}[b]{0.45\textwidth}
       \includegraphics{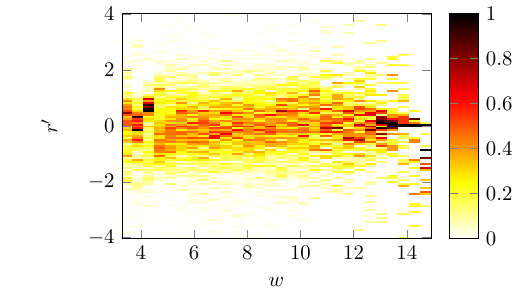}
	\caption{Wind support $[3.5,15]$}
	\label{fig:RescaledResidualDistributionRH}
    \end{subfigure}
    \caption{The conditional densities of the rescaled residuals}\label{fig:animals}
\end{figure}

Figure~\ref{fig:RescaledResidualDistributionFH} depicts the conditional distribution of the rescaled residual samples in the full wind speed range, while Figure~\ref{fig:RescaledResidualDistributionRH} concentrates on the range $[3.5, 15]$. There are some remarkable features that should be emphasized.

In Figure~\ref{fig:RescaledResidualDistributionFH}, it is notable that the rescaled residuals outside the range defined by the cut-in speed and the rated speed have distinctive patterns. These patterns are partly caused by the fact that the data are quantized and partly caused by the fact that we are rescaling the residuals, which in this case means that we are dividing with a standard deviation close to zero, see Figure~\ref{fig:ResidualVariances} for the values of  the standard deviation of the residuals.

In Figure~\ref{fig:RescaledResidualDistributionRH}, it is notable that the rescaled residuals have very similar conditional densities for different wind speed values between the cut-in speed and the rated speed, i.e. in the range $[3.5,15]$ the conditional distribution of the residuals given the wind speed is (approximately) independent of the wind speed. Thus, the wind speed is (approximately) independent from the rescaled residuals. This indicates that the WTPC captures most of the wind dependence in this range except for the variance. 

\begin{figure}[h!]
    \centering
    \begin{subfigure}[b]{0.45\textwidth}
	\includegraphics{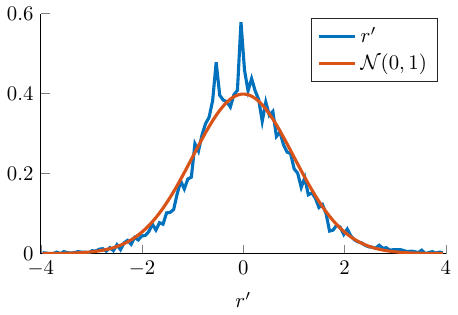}
	\caption{The marginal density (blue) of the rescaled residuals conditioned on $w\in[3.5, 15]$ and the density of the standard Gaussian distribution (red)}
	\label{fig:RescaledResidualDensityAndNormal}
    \end{subfigure}
\\
    \begin{subfigure}[b]{0.45\textwidth}
	\includegraphics{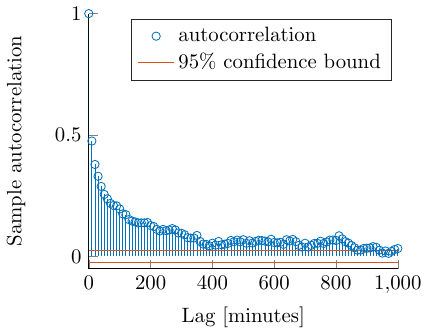}
	\caption{The sample autocorrelation of the rescaled residual \\    ~ \\     ~}
	\label{fig:RescaledResidualAutocorrelation}
    \end{subfigure}
    \caption{The marginal density and the autocorrelation of the rescaled residuals $r'_t$ }
\end{figure}

As it can be seen in Figure~\ref{fig:RescaledResidualDensityAndNormal}, the empirical density of the rescaled residuals inside the restricted wind speed interval resembles a Gaussian density. Combining this result with the fact that the residual is independent of the wind speed between the cut-in speed and the rated speed, we can reasonably assume that the rescaled residuals can be modeled using a Gaussian stochastic process.
Motivated by this result, in Section~\ref{sec:ConcreteARMAmodels}, we model the rescaled residuals $r'$ using an autoregressive moving-average (ARMA) model. This model can capture the inherent autocorrelation observed in the data,  cf. Figure~\ref{fig:RescaledResidualAutocorrelation}.

\subsection{Dynamic modeling}
\label{sec:ConcreteARMAmodels}
In all previous sections, for reasons of simplicity and for the clarity of the exposition, we suppressed the time index from the notation. This was also in accordance with the static models we investigated for modeling the WTPC. In this section, we further improve on the previous static models by incorporating a dynamic aspect satisfying also the Gaussian behavior of the rescaled residuals. This will additionally permit us to increase the short-term forecasting potential of our model. To this purpose, we reinstate the subscript $t$ (indicating the time dependence) to all variables and consider for modeling purposes an ARMA$(q_1,q_2)$ model, with
$q_1$ autoregressive terms (with coefficients $a _{i}$, $i=1,\ldots,q_1$) and $q_2$ moving-average terms  (with coefficients $c _{i}$, $i=1,\ldots,q_2$), i.e.,
\begin{equation}
\label{eq:ARMAform}
r'_{t}=1+\varepsilon _{t}+\sum _{i=1}^{q_1}a _{i}r'_{t-i}+\sum _{i=1}^{q_2}c _{i}\varepsilon _{t-i}.
\end{equation}
Taking into account the above model for the rescaled residuals, in the sequel, based on historical data till a given time $\tau$, we develop the forecasting model for the power output
\begin{equation}
\label{eq:DynamicPowerModel}
p_{t} = \mathcal{F}_{\bm{\theta}^{(f)}}(w_t, \phi_t, T_t) + \sigma_{w_t}{r}'_{t},\ t\geq\tau,
\end{equation}
where the estimated parameters for the model $\mathcal{F}_{{\bm{\theta}}^{(f)}}(w_t, \phi_t, T_t)$, the conditional standard deviation $\sigma_w$, and the parameters of the ARMA model are estimated based on historical data (e.g. same season in previous year), while the estimated values of the driving noise $\varepsilon_t$ depend on observations preceding $t$, but close to it in time, so this cannot be constructed based on historical data only.

We should note that the rescaled residuals seem to have a Gaussian density only inside the interval $[3.5,15]$ of the wind speed values. Outside of this range of values, the variance of the residuals is constant and seemingly very small. Thus, similarly to Section~\ref{sec:IncorporatingBounds}, we define the constrained model with regard to the rescaled residuals as follows
\begin{equation}
\label{eq:DynamicPowerModelDefinition}
\begin{split}
p_t =&  \mathcal{F}_{\bm{\theta}^{(f)}}(w_t, \phi_t, T_t)   \\
     &+ \I\{g_\ell \leq w_t \leq g_u \}  \sigma_{w_t}r'_t  \\
     &+ (1-\I\{g_\ell \leq w_t \leq g_u \})e_t,
\end{split}
\end{equation}
with $3.5\leq g_\ell<g_u\leq 15$. The values $g_\ell$ and $g_u$ will be determined, so as to ensure that within the  interval $[g_\ell,g_u]$ the conditional distribution of the rescaled residuals is Gaussian. Furthermore, $e_t$ is a Gaussian noise source, such that for every finite set of indexes $\{t_{1},\ldots ,t_{N}\}$ the corresponding components are independent and identically normally distributed. Moreover, $e_t$ is independent from the driving noise behind the rescaled residual signal,  $\varepsilon_t$.

It needs to be mentioned that the noise $\varepsilon_t$ cannot be obtained from the data using the model given in \eqref{eq:DynamicPowerModelDefinition}. To overcome this issue, we assume that the noise $\varepsilon_t$ and the rescaled residual process starts at zero at the beginning of the measurement time line, i.e. $\varepsilon_t$ and $r'_t$ are zero, for $t<0$, and assume that it stays ``frozen'' while the wind is outside the interval $[g_\ell,g_u]$. For the latter, we equivalently glue together consecutive periods of time for which the wind is within the desired range $[g_\ell,g_u]$.

The parameters of the constrained model described in Equation \eqref{eq:DynamicPowerModelDefinition} are:
\textit{i)} the parameters of the $\mathcal{F}_{\bm{\theta}^{(f)}}(w_t, \phi_t, T_t)$ model;
\textit{ii)} the parameters of the wind speed dependent residual rescaling factors $\sigma_{w}$;
\textit{iii)} the parameters of the ARMA($q_1,q_2$) model, say $\bm{\theta}_{\mathrm{ARMA}}$:  $\{a_1,\ldots,a_{q_1}\}$ and $\{c_1,\ldots,c_{q_2}\}$.
So the full parameter vector $\bm{\theta}$ of the model consists of $\bm{\theta}^{(f)}$, $\sigma_{w}$, $\{a_1,\ldots,a_{q_1}\}$ and $\{c_1,\ldots,c_{q_2}\}$.

As it is visible in Figure~\ref{fig:RescaledResidualDistributionFH}, the conditional distribution of the rescaled residuals conditioned on the wind speed is not Gaussian on the full wind speed support. In order to determine the lower and upper limits of the Gaussian range, $g_\ell$ and $g_u$, we performed the Anderson-Darling test to find the $p$-values for the conditional distributions that show how likely it is that the rescaled residuals, given the wind speed values, are samples from a standard normal distribution. Figure~\ref{fig:GaussianityTestForRescaledResiduals} shows the $p$-values of the test, along with the wind speed value boundaries that we used for later calculations. In particular, with relatively high confidence ($p\text{-value}>0.05$) we cannot reject the hypothesis that the samples come from a standard normal distribution for $g_\ell = 5.4$ and $g_u=13.6$, while outside these bounds the hypothesis can be rejected with extremely high confidence ($p\text{-value}\approx 0$).

\begin{figure}[h!]
	\centering
	\includegraphics{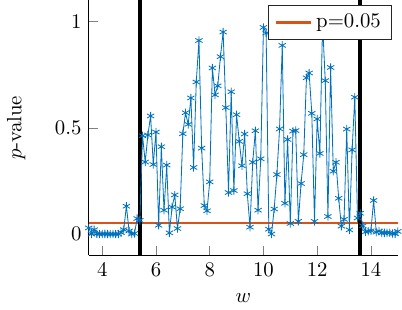}
	\caption{The $p$-value of the Anderson-Darling test for standard normality performed on the conditional distribution of the rescaled residuals}
	\label{fig:GaussianityTestForRescaledResiduals}
\end{figure}

Having decided on the values $g_\ell$ and $g_u$, we describe below the procedure for estimating the parameters $\bm{\theta}$ of the model described in Equation \eqref{eq:DynamicPowerModelDefinition}.
First, we estimate the parameters of the $\mathcal{F}_{\bm{\theta}^{(f)}}(w_t, \phi_t, T_t)$  model, next we estimate the parameters  $\sigma_{w}$ of the rescaled residuals, last we estimate the parameters of the ARMA model. This is indicated by the data, as the distribution of the rescaled residuals $r'_t$ is symmetric  around zero and it is independent of the wind $w_t$. All in all, this procedure is mathematically described as
\begin{equation}
\label{eq:DynamicModelSerialEstimation}
\hat{\bm \theta} = \argg_{\bm\theta} \left( \min_{\bm\theta_{\textrm{ARMA}}} \min_{\sigma_w} \min_{\bm{\theta}^{(f)}} \frac{1}{N} \sum_{k=1}^N (p_k-\hat{p}_k)^2\right),
\end{equation}
where the prediction $\hat{p}_k$ is obtained in accordance to Equation \eqref{eq:DynamicPowerModelDefinition} and $N$ is the total number of observations. Note that as $N\to\infty$, $\hat{\bm \theta} $ converges to the least square estimate obtained by optimizing for every parameter simultaneously, instead of the proposed sequential optimization. This follows from the fact that the residuals $p_k-\mathcal{F}_{\bm{\theta}^{(f)}}(w_t, \phi_t, T_t)$ are symmetric around zero.

\begin{table*}[t!]
    \centering
    \caption{ARMA model parameters belonging to different orders}
    \label{tab:ArmaModels}
     \begin{tabular}{ccccccc}\hline
        \bf{ Model order}&& \multicolumn{5}{c}{\bf{ model parameters}} \\  \hline
         $q_1=0, q_2=5$  && $c_1 = ~0.4188$ & $c_2 = ~0.2941$ & $c_3 = ~0.2379$ & $c_4 = ~0.1738$ & $c_5 = ~0.1085$ \\
         $q_1=5, q_2=0$  && $a_1 = ~0.4125$ & $a_2 = ~0.1271$ & $a_3 = ~0.0824$ & $a_4 = ~0.0375$ & $a_5 = ~0.0248$  \\
\multirow{ 2}{*}
        {$q_1=5, q_2=5$} && $a_1 = ~1.3982$ & $a_2 = -0.3649$ & $a_3 = -0.1556$ & $a_4 = ~0.3187$ & $a_5 = -0.2043$ \\
                         && $c_1 = -0.9894$ & $c_2 = ~0.0847$ & $c_3 = ~0.1463$ & $c_4 = -0.2927$ & $c_5 = ~0.0901$ \\ \hline
     \end{tabular} ~ \\[2ex]
\hrulefill
\vspace*{4pt}
\end{table*}

The estimated parameters $\hat{\bm\theta}$ are calculated as follows: For $\hat{\bm{\theta}}^{(f)}=(\hat{\bm{\theta}}^{(s)},\hat{c}_\phi,\hat{c}_T)$, we use
 the estimated parameters for the B-spline model  $\hat{\bm{\theta}}^{(s)}$ given in Section~\ref{sec:Spline}, and the environmental coefficients  $\hat{c}_\phi$ and $\hat{c}_T$  given in Section~\ref{sec:MorePhysicalParamaters}. For the rescaled residuals, $\hat{\sigma}_{w}$ is estimated in a non-parametric  way for every appearing wind value in the dataset and the  values are shown in Figure~\ref{fig:ResidualVariances}. For the parameters of the ARMA model, we refer to Table~\ref{tab:ArmaModels}. For the calculation of the parameters, we use the System Identification Toolbox \citep{Ljung2010SysIdToolbox} of Matlab.

\subsection{Comparison of model's forecasting capabilities}
\label{sec:forecasting}
This section describes the forecasting capabilities of the models outlined above. In particular, we use the simple B-spline WTPC model, cf. Equation \eqref{eq:SplinePowerModelNoBoundary}, as a baseline to underline the
improvement offered by utilizing the additional environmental regressors, cf. Equation \eqref{eq:AdditionalPhysicalLayers}, as well as modeling the variance and the correlations remaining in the residuals of the model, cf. Equation \eqref{eq:DynamicPowerModelDefinition}. The corresponding parameters for the models under consideration are presented in Sections~\ref{sec:Spline},~\ref{sec:ConcreteARMAmodels}, and~\ref{sec:MorePhysicalParamaters}, respectively. We depict the MSE of the three models in Figure~\ref{fig:predictionCapabilitesWithHorizons} as a function of the forecasting horizon.
For the calculation of the MSE as a function of the forecasting horizon, we need to note that although the sampling frequency of the data is every $\delta=10$ minutes, the forecasting horizon can receive any positive continuous value. Keeping this in mind, we define the MSE, given the forecasting horizon $h$, as follows
\begin{equation}
\label{eq:MSEForPredicionHorizon}
\mathrm{MSE}_h = \frac{1}{N-\lceil h/\delta \rceil} \sum_{k=\lceil h/\delta \rceil}^N \left(p_{k}- \hat{p}_{k|k-\lceil h/\delta \rceil}\right)^2,
\end{equation}
where for the prediction of the $k$-th value, $\hat{p}_{k|k-\lceil h/\delta \rceil}$, it is required to provide as an input the wind speed values $w_1,w_2,\ldots,w_k$, the temperature values $T_1,T_2,\ldots,T_k$, the angle values $\phi_1,\phi_2,\ldots,\phi_k$, and the power output values $p_1,\ldots,p_{k-\lceil h/\delta \rceil}$, i.e., we predict from the $k-\lceil h/\delta \rceil$ power output values the future, given perfect future information of the explanatory values.

\begin{figure}[h!]
    \centering
    \begin{subfigure}[b]{0.45\textwidth}
	\includegraphics{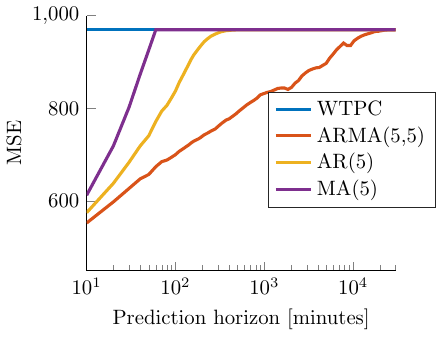}
	\caption{Simple WTPC}
	\label{fig:PredictionCapabilitesWithHorizonsValidationSpline}
    \end{subfigure}
    ~~~~~~
    \begin{subfigure}[b]{0.45\textwidth}
       \includegraphics{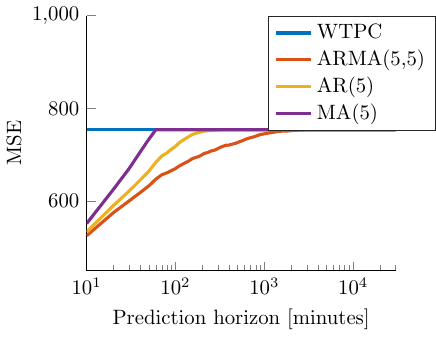}
	\caption{WTPC enhanced with temperature and angle}
	\label{fig:PredictionCapabilitesWithHorizonsValidationPhys}
    \end{subfigure}
    \caption{The MSE of the different models with respect to the prediction horizon $h$ evaluated on the validation set}
    \label{fig:predictionCapabilitesWithHorizons}
\end{figure}

As shown in Figure~\ref{fig:predictionCapabilitesWithHorizons}, the two static  WTPC models, cf.  Equations \eqref{eq:SplinePowerModelNoBoundary} and \eqref{eq:AdditionalPhysicalLayers}, given the wind speed, angle and temperature, have a constant MSE regardless of the forecasting horizon. Contrary, the dynamic ARMA$(q_1,q_2)$ models permit a significant reduction of the MSE, especially for short-term forecasting in the range of 1 to $50$, 1 to $10^2$,  or 1 to $10^3$ minutes, depending on the values of the $(q_1,q_2)$ parameters. Naturally, as the prediction horizon increases the added benefit of knowing the power output values from the past is getting less and less valuable. Furthermore, the least effective is the moving-average MA($q_2$) model structure, since it utilizes information only from the past $q_2$ samples. So if the prediction horizon reaches this limit, no past information is used. As the unconditional expectation of the zero mean Gaussian process is zero, the expectation of the rescaled residuals will also be zero. This can be seen in Figure~\ref{fig:predictionCapabilitesWithHorizons}, in which the MSE value of the MA$(5)$ model becomes equal to that of the corresponding static WTPC model for prediction horizons $h > 5\cdot\delta = 50$ minutes. While, for the same value of the forecasting horizon, the ARMA$(5,5)$ model has a reduced MSE value by  17\% compared to the corresponding static WTPC model.
In contrast to the short-term forecasting characteristics, the advantages of the dynamic models versus the static models disappear in the long-term. Thus, the dynamic WTPC model with the ARMA layer can be used for both short-term and long-term forecasting, as for short horizons it outperforms the static WTPC, while for longer horizons has the same performance as the static counterpart model.
This result is evident in light of Figure~\ref{fig:predictionCapabilitesWithHorizons} as the various MSE values of the dynamic models converge to the MSE value of the corresponding static WTPC model.

We need to note that since the autocorrelation of the rescaled residuals (calculated using the enhanced WTPC model), cf. Figure~\ref{fig:RescaledResidualAutocorrelation}, is significant for a long period of time (more than 400 minutes), estimating higher order ARMA models would reduce the MSE value in comparison to the corresponding static model but this effect will vanish for time horizons longer than the autocorrelation length of the rescaled residuals.

Comparing Figures~\ref{fig:PredictionCapabilitesWithHorizonsValidationSpline} and \ref{fig:PredictionCapabilitesWithHorizonsValidationPhys}, we note that enhancing the static model with the wind direction and the ambient temperature has two effects: the MSE drops significantly, but the added benefit of the dynamic layer vanishes faster ($10^3$ minutes instead of $10^4$). This is due to the fact that the estimated ARMA model in the case of the simple WTPC was additionally trying to capture the autocorrelation structure of the temperature, which is persistent for lengthy lags. While, in the enhanced model, the temperature is provided as a regressor and therefore the ARMA model only needs to capture the remaining residual, whose autocorrelation vanishes for lower lag lengths in comparison to the temperature.

\subsection{Forecasting confidence}
\label{sec:forecasting_confidence}
In this section, we are interested in investigating the performance of the dynamic model in terms of its forecasting ability. To this purpose, we visualize a power output trajectory in Figure~\ref{fig:powerPredictionTimeline} and depict the difference between the prediction and the actual measurements in Figure~\ref{fig:powerPredictionTimelineDifference}. In both figures, we define time 0 to be the starting point of the forecasting horizon and we assume that the wind speed, temperature and relative angle values are known also during the forecasting period, while the power output values are known only till time 0.
For the creation of the figures, we consider both the static and the dynamic WTPC model and plot their predictions together with the corresponding prediction intervals. The prediction interval of the static WTPC model has a constant width, while the dynamic model has a varying width depending on the value of the wind speed.  The 95\% confidence band for predicted power production trajectory $\hat{p}_{k|0}$ was calculated based on the ARMA model and the wind dependent scaling factor.

For the static model, the fixed width interval is a result of calculating the variance of the residuals using the static WTPC model with the $p-\hat{p}=p-\mathcal{F}_{\hat{\bm{\theta}}^{(f)}}(w, \phi, T)$, cf. \eqref{eq:AdditionalPhysicalLayers}, as two times the standard deviation of the residuals can be used as an approximation for the 95\% confidence region of the prediction. This calculation on our data results in a standard deviation for the residuals equal to $12.4925$. However, for this to hold it should be the case that the residuals are normally distributed and independent of the wind speed, but as we have already shown  this is not the case.
Note that the variance of the residuals for the static model is calculated using observations covering the support of the wind speed values $[0,25]$, thus simultaneously taking into account the part for which the static model is very accurate, $[0,g_\ell)\cup(g_u,25]$, and the part in which it is highly inaccurate, $[g_\ell,g_u]$. As a result, the estimate for the variance of the residuals is overly conservative in $[0,g_\ell)\cup(g_u,25]$, while it seems to underestimate the variance in $[g_\ell,g_u]$. This is clearly visible in Figure~\ref{fig:powerPredictionTimelineDifference} as during the first half of the forecasting horizon the wind speed was in $[0,g_\ell)\cup(g_u,25]$, while in the second half it is $[g_\ell,g_u]$.

\begin{figure}[h!]
  \begin{subfigure}[b]{0.9\linewidth}
	\includegraphics{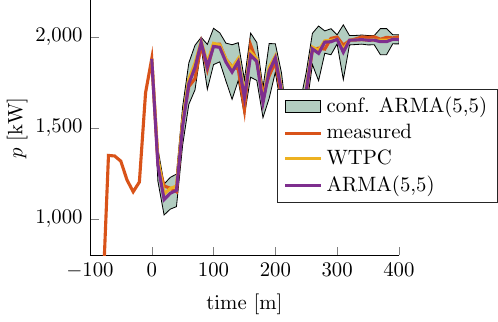}
	\caption{Typical prediction trajectory with the corresponding 95\% confidence bound}
	\label{fig:powerPredictionTimeline}
	\end{subfigure}\\
\begin{subfigure}[b]{0.9\linewidth}
	\includegraphics{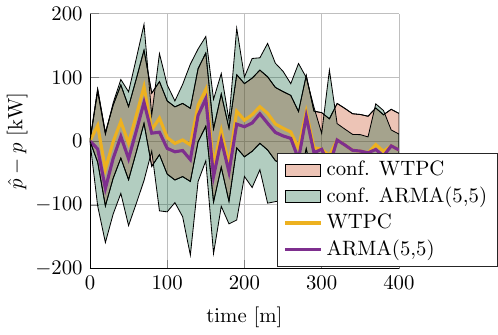}
	\caption{The difference between the prediction and the actual measured power}
	\label{fig:powerPredictionTimelineDifference}
    \end{subfigure}
    \caption{Forecasting capabilities of the static (WTPC) model and the dynamic model with the ARMA layer}
\end{figure}

For the dynamic model, cf. \eqref{eq:DynamicPowerModelDefinition}, the variance of the power estimate $\hat{p}_{k|k-\lceil h/\delta \rceil}$ can be calculated by considering the unknown random variables of $(\varepsilon_\ell)_{k-\lceil h/\delta \rceil < \ell \leq k}$, which drive the rescaled residual stochastic process $r'_{k|k-\lceil h/\delta \rceil}$, see \eqref{eq:ARMAform}.
We know that the variance of $r'_{k|k-\lceil h/\delta \rceil}$ is monotonically increasing with the prediction distance $\lceil h/\delta \rceil$ and it has a bounded limit, since the ARMA model is stable
(every solution $z \in \mathds{C}$ of the $1-\sum_{i=1}^{q_1} a_i z^{-1}=0$ characteristic equation has absolute value less than 1).
If we assume that $\varepsilon_t$ is an i.i.d. Gaussian signal, then the variance of the power predictions can be calculated from the ARMA model and the wind dependent rescaling $\sigma_{w_t}$.
The fact that the dynamic model results in a confidence band with varying width is initially surprising as the width might even shrink in size over time. The explanation for this result is that the proposed dynamic model contains a wind speed dependent scaling. This scaling factor has very small uncertainty when the wind speed is  outside the interval $[g_\ell, g_u]$, i.e., for wind speed values  outside the interval $[g_\ell, g_u]$, the variance is smaller in comparison to the corresponding  value calculated over the full support.
As it can be seen in Figure~\ref{fig:powerPredictionTimeline}, the wind value region, in which the width of the confidence band shrinks corresponds to wind speed values above the wind value $g_u$ (the wind speed is not depicted directly but this can be inferred from the power curve and the shown measured power). In this region of wind values, predictions are more accurate and the  confidence band becomes narrower than the corresponding confidence band of the static  WTPC model. While, for wind values inside the interval $[g_\ell, g_u]$, it can be seen that the confidence band of the dynamic model gets wider and may even contain the confidence interval of the static WTPC model. The increased width of the confidence band is due to the combined effect of predicting values further into the future as well as the changes in the wind speed, that also effect the rescaling factor.

The confidence band of the static WTPC model contains 100\% of the samples. To illustrate how imprecise the uncertainty estimate belonging to the WTPC model is, we can calculate the maximal confidence level for which the empirical confidence is not 100\%. The 41.0701\% confidence band contains 99.9920\% percent of the validation data, which shows a significant underestimation of uncertainty.
The 95\% confidence band corresponding to the dynamic model contains 96.0990\% of the samples in the validation time line. This shows that the uncertainty of the predictions can be evaluated quite reliably using the dynamic model, as the empirical confidence level is relatively close to the theoretical one.

\section{Concluding remarks}
\label{sec:conclusions}
The paper focuses on the short- and long-term power output forecast of a wind turbine based on past measurements of the wind speed, power output and other environmental factors, as well as perfect knowledge of the future wind speed, angle, and ambient temperature. We showed that the parametrization of the WTPC is not a key factor in improving prediction performance. The reason behind this is that for any  model with a sufficiently rich structure (in the case of non-parametric models that implies sufficiently high complexity) we can achieve MSE values close to the lower bound, as long as we have sufficient data to estimate all the unknown parameters at hand. This is not a problem in our case, as we have a trove of data, making it more important to consider models that have a rich enough structure and the unknown parameters can still be estimated with high numerical accuracy, e.g. the polynomial based models suffer from numerical instability issues. Given the available data we have at our disposal, our conclusion is that the B-spline model with a sufficiently high number of knots provides a good modeling choice, as it can capture every detail of a WTPC and it can be estimated in a numerically stable way. Of course, if the data were sparse but the wind speed still covered the range of [\SI[per-mode=symbol]{3.5}{\meter\per\second}, \SI[per-mode=symbol]{15}{\meter\per\second}], a better option would be a logistic model as such models maintain the shape of the WTPC.

The error between the actual power generation  values and those predicted by the WTPC have special characteristics that open up possibilities for better modeling. Below the cut-in speed and above the rated speed of the turbine the predictions are quite accurate. However, in the middle range of the wind speed this is no longer true. We have shown that, on the dataset at hand, a proper rescaling of these residuals can transform the residual signal into a Gaussian signal. Modeling this Gaussian rescaled residual as a stochastic time series allowed us to improve significantly the predictions. The proposed model structure was able to improve up to 40\% in predicting the power output  of the wind turbine on short-term predictions, while the long-term prediction capabilities of the model are identical to that of the WTPC.

\section*{Acknowledgments}
\noindent 
The authors acknowledge the Daisy4offshore consortium for the provision of the data. 
The work of S\'andor Kolumb\'an and Stella Kapodistria is supported by NWO through the Gravitation-grant NETWORKS-024.002.003. S\'andor Kolumb\'an also received funding from the Institute for Complex Molecular Systems (ICMS) in Eindhoven.  Furthermore, the research of Stella Kapodistria was partly done in the framework of the TKI-WoZ: Daisy4offshore project.

\bibliographystyle{apa}
\bibliography{paper01}

\end{document}